\documentclass[acmsmall,screen]{acmart}
\setcopyright{rightsretained}
\acmDOI{10.1145/3674657}
\acmYear{2024}
\acmJournal{PACMPL}
\acmVolume{8}
\acmNumber{ICFP}
\acmArticle{268}
\acmMonth{8}
\acmSubmissionID{icfp24main-p96-p}
\usepackage[utf8]{inputenc}
\usepackage{color}
\usepackage{xcolor}
\usepackage{amsmath}
\usepackage{bussproofs}
\usepackage{amsthm}
\usepackage{wasysym}
\usepackage{siunitx}
\usepackage{float}

\usepackage{tikz}
\usepackage{tikz-cd}

\usepackage{subcaption}

\newcommand{\vquad}{\vspace{10pt}}
\newenvironment{bprooftree}
  {\leavevmode\hbox\bgroup}
  {\DisplayProof\egroup}

\usepackage{listings}
\definecolor{commentgreen}{RGB}{2,112,10}
\definecolor{eminence}{RGB}{108,48,130}
\definecolor{weborange}{RGB}{255,165,0}
\definecolor{frenchplum}{RGB}{129,20,83}

\newif\ifdraft
\newif\ifext
\draftfalse

\newcommand{\refinement}[3]{{\{#1:#2\ |\ #3\}}}
\renewcommand{\models}[2]{{#1\ {\sf models}\ #2}}
\newcommand{\inttype}{\sf{int}}
\newcommand{\floattype}{{\sf float}}
\newcommand{\hd}[1]{{\sf hd}(#1)}
\newcommand{\tl}[1]{{\sf tl}(#1)}
\newcommand{\delay}[1]{{\sf delay}(#1)}
\newcommand{\fby}[1]{\ {\sf fby}_{#1}\ }
\renewcommand{\ifthenelse}[3]{{\sf if}\ {#1}\ {\sf then}\ {#2}\ {\sf else}\ {#3}}
\newcommand{\true}{{\sf true}}
\newcommand{\false}{{\sf false}}

\newcommand{\letrec}[2]{{\sf let\ rec}\ {#1}\ {\sf in}\ {#2}}
\newcommand{\dom}[1]{{\sf dom}({#1})}
\newcommand{\prev}[1]{{\sf prev}({#1})}

\newcommand{\hist}{h} 

\newcommand{\alw}{\square}

\newcommand{\nxt}{\ocircle}

\newcommand{\kw}[1]{{\sf #1}}
\newcommand{\nil}{\kw{nil}}
\renewcommand{\phi}{\varphi}

\newtheorem{theorem}{Theorem}
\newtheorem{lemma}{Lemma}
\newtheorem{defn}{Definition}

\renewcommand{\kw}[1]{{\sf #1}}
\renewcommand\xrightarrow[1]{\stackrel{#1}{\hookrightarrow}}

\title{Synchronous Programming with Refinement Types (Extended)}
\date{June 2024}
\newif\ifanon
\anonfalse

\author[J Chen]{Jiawei Chen}
\orcid{0000-0002-5461-6711}
\affiliation{%
  \institution{University of Michigan}
  \city{Ann Arbor}
  \country{USA}
}
\email{chenjw@umich.edu}

\author[J L Vargas de Mendonça]{José Luiz Vargas de Mendonça}
\orcid{0000-0002-1576-2492}
\affiliation{%
  \institution{University of Michigan}
  \city{Ann Arbor}
  \country{USA}
}
\email{joselvdm@umich.edu}

\author[B S Ayele]{Bereket Shimels Ayele}
\authornote{Work performed while at the University of Michigan}
\orcid{0000-0002-0095-2146}
\affiliation{%
  \institution{Addis Ababa Institute of Technology}
  \city{Addis Ababa}
  \country{Ethiopia}
}
\email{bereket-shimels.ayele1@louisiana.edu}

\author[B N Bekele]{Bereket Ngussie Bekele}
\authornotemark[1]
\orcid{0000-0002-3537-3988}
\affiliation{%
  \institution{Addis Ababa Institute of Technology}
  \city{Addis Ababa}
  \country{Ethiopia}
}
\email{bereket.bekele@mail.utoronto.ca}

\author[S Jalili]{Shayan Jalili}
\orcid{0000-0003-1535-5579}
\affiliation{%
  \institution{University of Michigan}
  \city{Ann Arbor}
  \country{USA}
}
\email{sjalili@umich.edu}

\author[P Sharma]{Pranjal Sharma}
\orcid{0000-0002-8217-173X}
\affiliation{%
  \institution{University of Michigan}
  \city{Ann Arbor}
  \country{USA}
}
\email{spranjal@umich.edu}

\author[N Wohlfeil]{Nicholas Wohlfeil}
\orcid{0009-0006-3220-6099}
\affiliation{%
  \institution{University of Michigan}
  \city{Ann Arbor}
  \country{USA}
}
\email{njwohlf@umich.edu}

\author[Y Zhang]{Yicheng Zhang}
\orcid{0000-0002-3134-1284}
\affiliation{%
  \institution{University of Michigan}
  \city{Ann Arbor}
  \country{USA}
}
\email{zyicheng@umich.edu}

\author[J-B Jeannin]{Jean-Baptiste Jeannin}
\orcid{0000-0001-6378-1447}
\affiliation{%
  \institution{University of Michigan}
  \city{Ann Arbor}
  \country{USA}
}
\email{jeannin@umich.edu}



\begin{document}
\begin{abstract}
Cyber-Physical Systems (CPS) consist of software interacting with the physical world, such as robots, vehicles, and industrial processes. CPS are frequently responsible for the safety of lives, property, or the environment, and so software correctness must be determined with a high degree of certainty. To that end, simply testing a CPS is insufficient, as its interactions with the physical world may be difficult to predict, and unsafe conditions may not be immediately obvious. Formal verification can provide stronger safety guarantees but relies on the accuracy of the verified system in representing the real system. Bringing together verification and implementation can be challenging, as languages that are typically used to implement CPS are not easy to formally verify, and languages that lend themselves well to verification often abstract away low-level implementation details. Translation between verification and implementation languages is possible, but requires additional assurances in the translation process and increases software complexity; having both in a single language is desirable. This paper presents a formalization of MARVeLus, a CPS language which combines verification and implementation. We develop a metatheory for its synchronous refinement type system and demonstrate verified synchronous programs executing on real systems.
\end{abstract}
\maketitle

\section{Introduction}

Much of the software we encounter, for example in automobiles, airplanes, or appliances, may interact with its physical environment, and thus are cyber-physical systems (CPS). Many CPS, such as those controlling vehicles or industrial processes, may pose great risk of harm in the event of a software failure. The development of safe CPS is therefore essential but is complicated by the difficulty of testing such systems, since their physical nature makes it difficult to exhaustively test all possible system and environmental configurations.

To this end, formal verification is an attractive approach that can be used to ascertain system safety. Hybrid automata \cite{goos1993hybrid} are commonly used to model CPS, though they do not directly represent executable programs. Other work in verifying CPS include Kind \cite{hagen2008scaling} and Kind 2 \cite{champion2016kind}, which generate invariants for discrete-time synchronous programs using an incremental, bounded SMT-based approach. 
Although formal verification is more thorough than testing, it carries its own set of challenges. Mainstream programming languages common in CPS and embedded software development are challenging to verify, as their semantics are often not well-defined or precise enough. On the other hand, languages primarily designed for formal verification might abstract away too many implementation-level details of the original system to generate executable programs \cite{platzer2008differential,goos1993hybrid}. Translation between verification and execution languages is possible but may introduce bugs of its own if unable to capture subtleties of the original system \cite{malecha2016towards}. Verified translation is possible \cite{bohrer2018veriphy}, though translation of any kind adds complexity to the development and debugging workflow.  Combining verification and execution into a single programming language avoids some of the compromises of either approach, giving CPS developers the assurance that the properties they verify  carry over to the real system.

The idea of languages that permit both verification and execution is not new; indeed, projects like Koord~\cite{ghosh2020koord}, VeriDrone~\cite{malecha2016towards} and VeriPhy~\cite{bohrer2018veriphy} offer verification and execution of CPS, with special consideration given towards ensuring run-time behavior reflects what was verified. In Koord, multi-agent robotics systems are modeled in a language with formal semantics that allows direct verification and simulation. Programs can execute via a middleware layer \cite{ghosh2020cyphyhouse} that replicates shared variables and interfaces with platform-specific controllers on real robots. System evolution follows a round-based approach that alternates between system and environment evolution, solving the common issue of synchronizing logical with physical time in a hybrid system model \cite{ghosh2020koord}. VeriPhy~\cite{bohrer2018veriphy} describes a method for constructing CPS models and specifications in differential dynamic logic (\kw{dL}) that can be soundly translated into a runtime monitor. This allows the safe mixing of a more optimal, but perhaps occasionally unsafe controller, with a safer alternative, though the method assumes such a safe controller exists. In comparison the present work focuses on compile-time verification of safety properties directly within an executable CPS language. 

One approach to compile-time verification is to encode desired properties in the language's type system. As a result of the Curry-Howard-Lambek isomorphism relating data types to propositions, a derivation of said augmented type directly implies the provability of the property. Dependent- and refinement-type systems such as that of F* \cite{swamy2016dependent} or Liquid Haskell \cite{jhala2020refinement}, enable type constraints to depend on program terms, and thus allow one to statically reason about a program's runtime behavior. A series of syntax-directed type-checking rules can then be used to generate constraints for a given program and its specification, the validity of which implies satisfaction of the specification.

This paper presents a formalization of the MARVeLus ({\bf M}ethod for {\bf A}utomated {\bf R}efinement-type {\bf Ve}rification of {\bf Lus}tre) type system. In existing work~\cite{chen2022synchronous}, typing rules for MARVeLus were proposed, along with an implementation of a verified robot collision avoidance controller with stationary obstacles. However, a proof of type safety had yet to be developed for this language. In this paper, we present the following enhancements to MARVeLus over existing work \cite{chen2022synchronous}:
\begin{enumerate}
    \item A formalization of MARVeLus, introducing more precise operational semantics based on comparable synchronous languages, the introduction of a subset of Linear Temporal Logic (LTL) trace predicates for type specifications\cite{colaco2006mixing,parissis1996test,caspi_synchronous_1996,ratel1992definition,caspi_co-iterative_1998}, and adjustments to the existing type system to enable a proof of type preservation;
    \item A definition and proof of type safety for the synchronous refinement type metatheory;
    \item Improved experimental methods supporting dynamic obstacles, such as collision avoidance when following a moving vehicle.
\end{enumerate}
We are primarily focused on the correct formalization of a type-safe executable synchronous programming language. Although existing works present proof automation \cite{goos_certifying_2001} or invariant generation techniques \cite{champion2016kind} for similar languages, we note that these developments are orthogonal to our present goals.

This paper is organized into five main parts. Section \ref{sec:2} introduces the MARVeLus language and presents a motivating example. Section \ref{sec:3} covers preliminary concepts and methods employed by MARVeLus for verification. Sections \ref{sec:4}, \ref{sec:5}, and \ref{sec:6} respectively describe the syntax, semantics, and type system that comprise our improved formalism for MARVeLus, including a definition for type preservation in synchronous programs. Sections \ref{sec:7} and \ref{sec:8} detail the implementation of MARVeLus atop an existing synchronous language. Finally, Sections \ref{sec:9} and \ref{sec:10} present the methods employed to execute MARVeLus on hardware, including an experiment with physical robots. We then conclude with a discussion of related work and future extensions.

\section{Overview}
\label{sec:2}

MARVeLus enables CPS developers to verify and deploy programs to real-world CPS. By leveraging the existing synchronous program simulation and compilation infrastructure of Zélus \cite{benveniste2011hybrid}, the language benefits from the great care already taken in ensuring its accuracy in modeling synchronous programs \cite{benveniste2014type-based,bourke2013zelus}. To this, we add automated verification of program properties expressed as type annotations inspired by refinement types \cite{rondon2008liquid,vazou2014refinement}. 

\subsection{A Simple Industrial Controller}
We present a simplified verification and implementation workflow to illustrate how one might use the MARVeLus extensions in designing CPS software. Suppose one wants to design an industrial controller responsible for maintaining the liquid level in a tank, similar to the scenario described by Kamburjan \cite{kamburjan2021post-conditions}. The tank has a maximum capacity of \SI{20}{L} and a fixed outflow rate of \SI{0.1}{L/s}. We require that the tank level stays within \SI{1}{L} and \SI{19}{L} to prevent the tank from completely emptying or overflowing. The controller switches a solenoid valve, which can only either be completely open or closed, and can deliver \SI{0.5}{L/s} when open. We design the controller so that it opens the valve when the tank level drops below \SI{1.5}{L} and closes the valve when it exceeds \SI{18.5}{L}, to afford a \SI{0.5}{L} safety margin. For the sake of simulation, we assume this is a discrete-time system, with a constant time step of \SI{0.1}{s}. 

\begin{figure*}
    \centering
\begin{lstlisting}[escapeinside={(*}{*)}]
let dt:{v:float | v > 0.} = 0.1
let maxflow:{v:float | v > 0.} = 0.5
let outflow:{v:float | v > 0.} = 0.1
let minlevel:{v:float | v > 0} = 1.
let maxlevel:{v:float | v > minlevel} = 19.
let node exec ()  : {(vf:float)*(vl:float) | true} =
  let rec ((flow, level):{(f:float)*(l:float) | 
      ((l >= minlevel && l <= maxlevel) && 
      ((l <= (maxlevel -. dt *. (maxflow -. outflow)) && f = maxflow) 
        || (l >= (minlevel +. dt *. outflow) && f = 0.)))}) =
    (0., 15.) fby 
    ((if (level <= (minlevel +. 0.5)) then maxflow else 
      (if (level >= (maxlevel -. 0.5)) then 0. else flow)),
      (level +. (flow -. outflow) *. dt) models (robot_get ("level"))) 
    in (() models robot_str ("flow", flow)); (flow, level)
\end{lstlisting}
    \caption{The industrial controller code with added refinements and external interfaces}
    \label{fig:tempverifiedcode}
 \end{figure*}

A verified MARVeLus implementation is shown in Fig. \ref{fig:tempverifiedcode}. Lines 1-5 define system constants such as the time step, maximum flow rate, and the tank's outflow rate. We refine these variables with type annotations, such as assuming they are positive, and that the maximum tank level must be greater than the minimum level. Lines 11-14 describe the state variables and their evolutions as a stream of tuples, of which each element represents the flow rate and liquid level at each instant. The type annotation on the tuple requires that both elements are floats, and that together they satisfy the conditions given in Lines 8-10. Note that we use $\texttt{f}$ and $\texttt{l}$ to refer to the first and second elements of the tuple to distinguish them from their instantiations in the remainder of the program. 

The safety condition is specified on Line 8, requiring that the tank level always remain between $\texttt{minlevel}$ and $\texttt{maxlevel}$. However, this specification alone is not sufficiently strong to verify the system, as it is not inductive. That is to say, proving the satisfaction of the safety specification now does not allow one to prove safety for the next time interval. Provided only this specification, the verifier produces counterexamples, even if they are not physically feasible. For instance, if the liquid level is at exactly \SI{19}{L} and the valve was open in the previous instant, the controller cannot respond quickly enough to close the valve before the tank overflows. However, we know that the system has started in a safe state, and that the controller will not have allowed the liquid level to reach \SI{19}{L} with an open valve. Therefore, we must specify an additional invariant (Lines 9-10) to assume that the controller will have made the correct decision in the previous instant. Specifically, we assume that the controller in the previous instant has left us with at least one time step to react and prevent an unsafe condition. By inspecting the controller code, one can assure themselves that this assumption is indeed satisfied. The verifier can then use this fact to check that the controller will continue to make the correct decision in the next instant, thus proving the property inductively.

Line 11 sets the initial condition of the system; we assume the tank level begins at \SI{15}{L} and the valve is initially closed. Lines 12-13 describe the controller's behavior and its effect on the commanded flow rate. The system reacts if the level approaches the limits and otherwise retains the previous flow rate. Line 14 then evolves the physical system, simulating the tank's liquid level based on inflow and outflow rates, or retrieving the actual tank level from a sensor labeled ``level''. We use the $\texttt{models}$ syntax to separate the simulation model (on the left) and the actual implementation (on the right); only the left side is verified, and we assume that the supplied model accurately reflects the real system dynamics. Line 15 then publishes the commanded flow rate to an actuator labeled ``flow'', and returns the flow rate and level to be used in other parts of the program.

\section{Preliminaries}
\label{sec:3}
\subsection{Synchronous Programming}
Synchronous programming treats all program data as streams which may take on different values at any given moment. Program computations are treated as manipulations to these streams to influence their evolution in time. Streams are well-suited for modeling embedded systems, which often handle streams of inputs and outputs in real time \cite{florence2019calculus}. For example, an autonomous vehicle must constantly interpret sensor data and produce control commands in a timely fashion to avoid an accident. Synchronous programming has gained significant traction in industry, where tools such as SCADE (a derivative of Lustre) \cite{colaco2017scade} and Esterel \cite{boussinot1991esterel} have been entrusted with designing integrated circuits, avionics, and other critical systems. Thus, a verified synchronous language has the advantage of familiarity for eventual industrial adoption. 

Synchronous programming observes logical time, in which a clock cycle is determined by certain actions or events involving the system, which may be decoupled from physical (``wall-clock'') time. All computations in a synchronous program must finish before the next clock cycle; thus, they are synchronized by this basic unit of time. This lends itself well to the design of embedded systems, as the data-flow model of synchronous programming closely resembles the flow of signals in electronic circuits~\cite{halbwachs1991synchronous} and imposes time constraints on computations. A complication that arises from combining software with the physical environment is the reconciliation of logical and physical clocks. In this paper, we focus only on discrete systems and reserve the examination of verifying systems spanning these two time domains for future work, noting that significant work has already gone into accurately simulating hybrid programs in Zélus \cite{benveniste2011hybrid}. We discretize the dynamics and clock our synchronous program at the rate of discretization. Since MARVeLus is primarily concerned with the types of streams and the values they emit, we assume that all streams have compatible clocks, a property already statically verified by the Zélus compiler. 

Another useful property of synchronous programs is that they operate in bounded memory \cite{bourke2013zelus,caspi1986lustre}, which limits the number of previous state values that are allowed to be observed at any time. This precludes, for instance, the arbitrary storage of sensor values, and makes memory usage more predictable. This is especially important for embedded systems, which often have limited computational resources. By construction, streams can only access a fixed number of previous states to derive the next state.

Lustre, and by extension Z\'elus, also observe a notion of causality, meaning that a program may only depend on values available in past cycles, and on present values in a way that does not produce circular dependencies. For instance, a stream may not instantaneously depend on its own value in the current cycle, but may depend on the value of another stream, provided the other stream does not already depend on it. Causality is a useful property for a CPS modeling language as it ensures systems created in the language are physically feasible. Causality is checked statically by the Z\'elus compiler and has already been formalized as a dedicated causality-analysis type system in existing work \cite{benveniste2014type-based,goos_modular_2001}, so we henceforth assume that programs we check already obey causality with respect to this formalism. To this end, we assume that a term $e$ is always well-typed in the causality type system \cite{goos_modular_2001}.

\subsection{Refinement Types}
Refinement types enhance the standard type system by allowing types to be specified using logical predicates that can take values of terms as inputs. These predicates serve as constraints on the base type such that all terms that inhabit the refined type satisfy the predicate. For instance, positive floating-point numbers can be specified using the refinement type $\refinement{w}{\floattype}{w>0.}$

Refinement types belong to the broader field of dependent types, which are, generally speaking, types that can depend on the values of terms \cite{pierce2005advanced}. Although full dependent type theory has the advantage of expressiveness~\cite{swamy2016dependent}, refinement types lend themselves well to generating verification conditions that depend on terms, while remaining decidable if restricted to the theories of equality, uninterpreted functions, and real arithmetic \cite{rondon2008liquid}. As a result, refinement types lose the expressivity of allowing arbitrary functions in refinements, abstracting away the bodies of non-arithmetic functions \cite{rondon2008liquid}, which if interpreted would jeopardize decidability \cite{jhala2020refinement}. 

\subsection{Verification Via Typing}
Verification via typing relies on the Curry-Howard-Lambek isomorphism to relate desired program properties (propositions) with types, and thus property checking becomes type checking \cite{pierce2005advanced}. Furthermore, type-checking is performed statically at compile-time, so verification at this stage in compilation has the chance to catch bugs in the program before any code is actually executed. This is particularly attractive to CPS, where testing on a real system can be difficult, risky, or costly. Verification via typing also allows our verifier implementation to be modular; all terms still possess a base type that can be checked using the existing type checker, while terms augmented with type refinements can provide additional information or conditions to the verifier. This allows the base Zélus type checker to remain independent of our additions. Refinement types \cite{jhala2020refinement} provide a method of adding annotations to the type system while maintaining decidability. An additional benefit of introducing verification through type-checking is that it promotes modularity in user code; terms that have been type-checked can then be re-used in other parts of the source code, carrying with them the proof of safety \cite{eichholz2022dependently-typed}. This can be especially important with large and/or long-term projects requiring extensive integration or regression testing of patches or updates.

\section{Syntax}
\label{sec:4}
\begin{figure}
    \centering
    \begin{align*}
        v~::=~&~c && \text{Literal Values}\\
              &|~(v_1,~v_2...,~v_n) &&\text{Tuples}\\
        \\
        e~::=~&~v &&\text{Values}\\
              &|~\kw{let}_\hist~(x:\tau)=e_1\kw{~in~}e_2 &&\text{Local Bindings}\\
              &|~\kw{let~rec}_\hist~(x:\tau)=e_1\kw{~in~}e_2 &&\text{Recursive Bindings}\\
              &|~f~x &&\text{Application}\\
              &|~e_1\kw{~fby~}e_2 &&\text{Stream Composition}\\
              &|~\delay{e} && \text{Delayed Evaluation}\\
              &|~(e_1,~e_2...,~e_n) &&\text{Tuples of Expressions}\\
              &|~\kw{if~}x_c\kw{~then~}e_t\kw{~else~}e_f &&\text{Conditionals}\\
              &|~e\kw{~models~}r &&\text{Modeling}\\
              \\
        r~::=~&~\kw{robot\_{get}~} k && \text{Robot Get}\\
              &|~\kw{robot\_{str}~} k~e && \text{Robot Store}\\
              \\
        d~::=~&~\kw{let~}f~(x:\tau_1):\tau_2=e &&\text{Named Function Definition}
    \end{align*}
    \caption{MARVeLus Syntax}
    \label{fig:streamsyntax}
\end{figure}

The stream syntax (Fig. \ref{fig:streamsyntax}) is modeled around Lustre streams, which forms a discrete-time subset of the existing Zélus syntax \cite{benveniste2011hybrid}. Similarly to other synchronous languages \cite{caspi_synchronous_1996}, we make the distinction between values $v$, which are constants emitted by stream expressions. We introduce a new entry $h$ (``history'') in recursive and non-recursive local bindings, which embeds a bound variable's past values into the syntax. This is essential for our operational semantics as stream expressions can depend on values that have occurred in the past and thus we may need to recover those values to re-construct the appropriate derivation tree. A syntactic sugar we employ is allowing local bindings to be defined without histories; these are considered uninitialized and are implicitly assigned empty histories ($h=[]$) Similarly to other refinement type systems \cite{rondon2008liquid,jhala2020refinement}, we require that certain constructs take variable arguments (such as function application and conditionals) to simplify the metatheory by preventing the introduction of arbitrary expressions into refinement predicates \cite{jhala2020refinement}. Another syntactic sugar we employ is to allow certain non-variable terms syntactically considered ``safe'' for the refinement logic to be used as argument functions (such as simple arithmetic expressions and conditionals), which greatly simplifies programs written in the real-world implementation of the language. We note that Zélus only supports named function definition at the top-level \cite{benveniste2011hybrid}, which we replicate.

Some elements of the syntax are not user-facing and are exclusively to support internal functionality. The $\delay{}$ operator is one such case, which allows an expression's evaluation to be deferred for one instant while preserving its original behavior despite environment updates. Likewise, we only allow \kw{let}-bindings (both recursive and non-recursive) without initialized histories (i.e. $h=[]$) in the user-facing portion of the language. This ensures prevents arbitrary values from being introduced to histories.

As part of the robotics extensions to Zélus, we introduce built-in keywords to facilitate reading and writing values to hardware ($r$) as shared variables labeled by string constants $k$. Additionally, we introduce the syntax $\models{e}{r}$ which allows one to define two implementations for a stream---a deterministic model for simulation and verification, and the implementation involving real hardware. As it is impossible to predict or verify at compile-time the exact outputs produced by real sensors, the CPS designer is responsible for ensuring that the model accurately represents actual behavior. In addition to the robotics extensions, we introduce the syntax necessary to support refinement types in Zélus.

\subsection{Specification Syntax}
The core of the specification syntax (Fig. \ref{fig:typesyntax}) is derived from similar refinement type systems \cite{jhala2020refinement, rondon2008liquid}. Typically, refinement type systems require that refinement predicates come from decidable logics such as the quantifier-free logic of equality, uninterpreted functions, and linear arithmetic. We note that allowing nonlinear polynomial real arithmetic maintains decidability, as shown by Tarski \cite{tarski_decision_1951}. We find this addition useful for specifying the kinematics of physical systems and it is supported by the Z3 solver \cite{hutchison2008z3}.

A type can either be a base type $b$, which represents a stream consisting solely of values belonging to a particular primitive data type, or refinement types. Refinement types of the form $\refinement{w}{b}{\phi}$ allow one to specify streams that satisfy a predicate $\phi$, and introduce a value variable $w$ which binds to any term inhabiting the type, allowing $\phi$ to reference terms \cite{rondon2008liquid}. For example, one may write the type of non-negative integers as $\refinement{w}{\inttype}{w\geq0}$. As a natural extension of standard refinement types, predicates in our refinement type specifications become streams of Boolean values. 

One of our main contributions over existing work \cite{chen2022synchronous} is extending the predicate syntax of MARVeLus to explicitly accommodate select temporal predicates from Linear Temporal Logic (LTL) \cite{pnueli1977temporal}. The refinement predicates discussed thus far form the syntax of ``state predicates'', which in the absence of temporal operators, reason about the immediate value of a stream at the current instant. To this, we add the temporal operators ``Always'' ($\alw$) and ``Next'' ($\nxt$), along with the conjunction of predicates containing these operators. This expanded set comprises the full set of specifications we allow as refinement predicates in the type syntax. While LTL normally includes other temporal operators along with disjunction at the top-level, we choose to focus primarily on safety properties which can be expressed globally (i.e. with $\alw$), which allows proofs to be inductive. In practice this is often sufficiently rich for describing most desirable behaviors of safety-critical CPS, as these properties typically require specifying that something either always or never happens.

\begin{figure}
    \centering
    \begin{align*}
        \tau~::=~&b~|~\{v:b~|~\phi\}~|~\tau\rightarrow\tau&&\text{Types}\\
        b~::=~&\kw{int}~|~\kw{float}~|~\kw{bool}~|~b_1\times b_2\times ...\times b_n &&\text{Base Types}\\
         p,q::= &\ \true \ |\ \false \ |\ x\ |\ e_1 = e_2\ |\ e_1 > e_2\ |\ p \land q\ |\ \lnot p&&\text{State Predicates}\\
        \phi, \psi ::=&\ p\ |\ \alw \phi \ |\ \nxt \phi\ |\ \phi \land \psi&&\text{Trace Predicates}
    \end{align*}
    \caption{Refinement Type Syntax, where we note that $e_1,e_2$ are terms within the logic of quantifier-free equality, uninterpreted functions, and arithmetic (QF-EUFA)}
    \label{fig:typesyntax}
\end{figure}

\section{Stream Semantics}
\label{sec:5}
MARVeLus programs operate on streams of data, producing values at each instant. Although the streams themselves do not terminate, they may emit a value at a given instant representing the stream's current state. Lustre's clock calculus allows streams that do not emit values on all clock cycles \cite{halbwachs1991synchronous}, but we currently consider only streams that are compatible with the basic clock, i.e. streams that always emit a value on every cycle of the most frequent clock. 

We describe the operational semantics of MARVeLus with the judgement $S;~\sigma\vdash e\xrightarrow{v} e'$, which represents the stream $e$ in stream context $\sigma$ and function context $S$ that simultaneously transitions into the new stream $e'$ and emits a value $v$ at each clock cycle. Similarly to other interpretations of Lustre semantics \cite{ratel1992definition}, the context $\sigma$ is comprised of mappings between variables and histories of their past values (Fig. \ref{fig:sigma-ops}). That is, for any $x\in\dom{\sigma}$, $\sigma(x)=\hist$ where $\hist=v^{-n}::v^{-n+1}...v^{0}$ is a history of the past values of $x$, up to the immediately preceding time instant. 
\begin{figure}
    \centering
    \begin{align*}
        \sigma~&::=~\varnothing~|~\sigma,x=\hist\\
        S~&::=~\varnothing~|~S, f(x)=e\\
        \hist ~&::=~[]~|~\hist::v\\
        \kw{prev}(\hist::v) &\triangleq \hist\\
        \kw{prev}(\sigma) &\triangleq \{\kw{prev}(\sigma(x))~| ~x\in \dom{\sigma}\}
    \end{align*}
    \caption{Definitions for function environment $S$, term environment $\sigma$, and the histories $h$ of variables bound in $\sigma$}
    \label{fig:sigma-ops}
\end{figure}

Consistent with the design of Lustre \cite{halbwachs1991synchronous}, loops and recursive function calls are excluded. As a result, the evaluation of a stream $e$ at any instant is terminating, meaning a value $v$ is always emitted on a transition of $e$.

One of the challenges in deriving operational semantics arose from reconciling the behavior of recursively-defined streams ($\kw{let~rec}$) with stream composition ($\kw{fby}$). In an earlier iteration of the semantics \cite{chen2022synchronous}, the left and right hand sides of ($\kw{fby}$) are evaluated simultaneously. When a stream depends on its past values, it necessarily employs \kw{fby} to both initialize the stream and delay self-references so that they do not result in unbounded recursion. For example, in the program $\letrec{x=0\fby{}x+1}{x}$, $x$ is defined initially as 0, and then for all subsequent instants is defined as the sum of 1 and the past value of $x$. This is not straightforward to characterize in a syntax-guided manner. Part of the reason is that, when evaluating $0\fby{}x+1$, the notion of the left-hand side becoming the initial value of $x$ is lost. This information is only known at the level of the recursive \kw{let}-binding. This is problematic if the left and right sides of $0\fby{}x+1$ must be evaluated together, because $x+1$ cannot be computed unless $x$ is known to be initially $0$, which is not known until $0\fby{}x+1$ is fully evaluated! In other (non-operational) semantics \cite{caspi_co-iterative_1998, caspi_synchronous_1996}, this was in part resolved by deferring the computation of the right-hand side of \kw{fby} (or its semantic equivalent). However, simply deferring the computation reveals a new problem: if the deferred expression contains free variables that have updated since the delay occurred, the behavior of this expression will change so that it no longer behaves like a properly delayed version of itself. Prior semantics address this by either storing intermediate values for delayed terms (akin to a buffer) \cite{caspi_co-iterative_1998}, or by leaving variables symbolic \cite{caspi_synchronous_1996}. Our operational approach takes inspiration from the former, in that local bindings maintain a history of their variables' past values so that deferred computations can refer back to the states of variables from the time before the delay occurred. Finally, one more issue to address is preserving this history in a syntactical manner. At each time instant, the operational semantics for the next step is derived independently as a new derivation tree. As such, the derivation needs to be re-populated with the past values of variables from the previous derivation so that the histories of variables can be recovered. This motivated the changes we made to the semantics of \kw{fby} and \kw{let}-bindings, and is essential for a syntax-directed operational semantics.

\subsection{Summary of Stream Behaviors}
Drawing inspiration from Lustre and Esterel \cite{ratel1992definition,parissis1996test,colaco2006mixing}, we define the operational semantics of MARVeLus in Fig. \ref{fig:streamsemantics}
\begin{enumerate}
    \item \textit{Constant Streams}: These are literals lifted from the base language; they return a fixed value each instant. Lustre treats all primitives as constant streams (i.e. \texttt{5} in Lustre is the constant stream that always emits 5)~\cite{caspi1986lustre}, and we adopt this convention as well.

    \item \textit{Variables}: Variables return the most recent value bound to them in $\sigma$. A variable $x$ steps to $x$ because variables themselves do not update; instead their entry in $\sigma$ is updated each instant to reflect the value of the stream to which it was originally bound; this update is performed by (S-LET) and (S-LETREC), the two places where new variables are bound.
    
    \item \textit{Pointwise Applications of Functions to Argument Streams}: Similarly to how a function can be mapped over an array, we define the pointwise application of a named function $f$ on an argument variable $x$ as $f~x$, the resultant stream being $f$ applied to the most recent value bound to $x$. Functions are imported from the non-stream base language and are assumed to be well-typed with respect to that language. The semantics are defined by the rule (S-APP).
    
    \item \textit{Compositions of Streams}: Lustre and Zélus allow a stream to be ``attached'' to another via $s {\sf ~fby~} t$ (pronounced ``followed-by''), resulting in a new stream consisting of the initial value of $s$ followed by the stream $t$ delayed by one clock, i.e. $s_0,~t_1,~t_2,...$. The behavior of ${\sf ~fby~}$ is defined by the rule (S-FBY). When the expression $e_1\kw{~fby~}e_2$ is evaluated, it emits the value that $e_1$ emits when it is evaluated, and then rewrites to $\delay{e_2}$. While intuitively one may think the result should simply be $e_2$, as the purpose of ${\sf fby}$ is to defer computation of the second computation, it is possible that free variables occur in $e_2$. This requires special treatment because the values of these free variables may have updated during the previous evaluation, and so $e_2$ evaluated in the current instant may not produce the same values it did in the previous instant. The intent of $\kw{~fby~}$ is to ``delay'' the second expression by one instant, so that it produces the same values it would have had it not been delayed, but one instant later. Essentially this ``shifts'' its behavior in time while also accounting for variables that may be assigned new values in the meantime. This operation could only be done through a special $\delay{}$ operator and corresponding (S-DELAY) semantics, which bridges the past and the present to allow expressions and their behaviors to be preserved in a newer context. 

    \item \textit{Delaying Streams}: Streams that are delayed by $\kw{fby}$ may contain variables which were bound outside the scope of the $\kw{fby}$. Consider the program in Fig. $\ref{fig:delaycode}$. Note how the variable $y$ is used inside the scope of a $\kw{fby}$ (in $x+y$), but $y$ is bound outside the scope of this delay. This means that $y$ is continually assigned new values by the evaluation of $0 \kw{~fby~} 1$, which is not under a delay. If $\kw{fby}$ simply deferred evaluation of $x+y$ without considering that variables will have updated in the meantime, the first evaluation of $x+y$ will have skipped the initial value of $y$, putting $x$ and $y$ out of sync with one another. Furthermore, $x$ will instantaneously depend on itself. Therefore, evaluation of $\kw{fby}$ is accompanied by $\kw{delay}$, which re-evaluates expressions in the previous context. Because previous values of variables in $\sigma$ are retained in $\kw{let}$ and $\kw{let~rec}$, the only places that bind new variables, it is possible to ``rewind'' $\sigma$ to previous states simply by removing the last value from each entry. In other words, the previous $\sigma$, $\prev{\sigma}$ is always completely recoverable from $\sigma$. This is encapsulated by $\kw{prev}()$ operator (Fig. \ref{fig:sigma-ops}). Variables that are delayed are guaranteed to have a value available in its corresponding context, since they are updated at least as many times as $\kw{delay}$ is applied through (S-FBY), since each delay will have been accompanied by an evaluation of the outer sub-expression which bound the variable. Therefore, it is impossible to revert $\sigma$ to a state before the ``beginning of time'' for a program. 

    \item \textit{Robot-Specific Expressions}: The ``robot expressions'' \kw{robot\_{get}}, \kw{robot\_{str}}, and \kw{models} serve to connect the rest of the semantics with real-world hardware and are largely semantically transparent, with the exception of combining $\kw{robot\_{get}}$ with $\kw{models}$ when a program is running in ``robot-mode''. In this mode, the program takes in real sensor values using $\kw{robot\_{get}}$. Since it is impossible to predict the values supplied by real-world sensors and because $\kw{robot\_{get}}$ can only be present on the right-hand side of $\kw{models}$, $\kw{robot\_{get}}$ as well as $\kw{models}$ expressions in robot-mode have no operational semantics. A $\kw{models}$ expression while not in robot-mode (as indicated by the side condition) is completely transparent to the underlying stream expression, and thus takes on identical semantics. $\kw{robot\_str}$ is assumed to always be a successful write operation performed on a piece of robot hardware, but has no other effect on the MARVeLus program. 

    \begin{figure*}
    \centering
    \begin{subfigure}{\textwidth}
    \centering
        \begin{verbatim}
            let rec x = 
              (let y = 0 fby 1 in (0 fby x+y))
              in x
        \end{verbatim}
        \caption{A program that contains variables used inside the scope of $\kw{fby}$}
        \label{fig:delaycode}
    \end{subfigure}
    \begin{subfigure}{\textwidth}
    \centering
        \begin{tikzcd}[column sep=small,row sep=small]
        x & 0 \arrow[r]  & 1 \arrow[r]  & 2 \arrow[r]  & 3 \arrow[r]  & 4 \arrow[r]  & 5 \arrow[r]  & \ldots \\
        y & 0  & 1 \arrow[u] & 1 \arrow[u] & 1 \arrow[u] & 1 \arrow[u] & 1 \arrow[u] & \ldots
        \end{tikzcd}
        \caption{Improperly delayed variables. The arrows represent the values of $x$ and $y$ that are used when computing the next value of $x$. Note how the first value of $y$ is skipped in computing $x+y$}
        \label{fig:delaywrong}
    \end{subfigure}

    \begin{subfigure}{\textwidth}
    \centering
        \begin{tikzcd}[column sep=small,row sep=small]
        x & 0 \arrow[r]  & 0 \arrow[r]  & 1 \arrow[r]  & 2 \arrow[r]  & 3 \arrow[r]  & 4 \arrow[r]  & \ldots \\
        y & 0 \arrow[ru] & 1 \arrow[ru] & 1 \arrow[ru] & 1 \arrow[ru] & 1 \arrow[ru] & 1 \arrow[ru] & \ldots
        \end{tikzcd}
        \caption{Properly delayed variables. No values of $y$ are skipped when computing $x+y$.}
        \label{fig:delayright}
    \end{subfigure}
        \caption{Traces of the variables $x$ and $y$ in a stream program (\ref{fig:delaycode}), if the variables are improperly (\ref{fig:delaywrong}) and properly (\ref{fig:delayright}) delayed}
        \label{fig:delaystreams}
    \end{figure*}
    
    \item \textit{Stream Recursion}: Streams can be defined recursively in terms of themselves. A stream $\letrec{x=e_1}{e_2}$ allows $x$ to occur in $e_1$ under special circumstances. Although recursively defined streams syntactically resemble recursive definitions in conventional languages, their semantics are entirely different. As with Lustre, MARVeLus considers variables to be identically equal to their definitions \cite{halbwachs1991synchronous}. Consequently, one must be mindful of circular dependencies in definitions. Streams must not be defined in terms of their current values, as one may otherwise produce contradictions ($\letrec{x=x+1}{x}$) or nondeterministic behavior ($\letrec{x=x}{x}$). Both scenarios have ``instantaneous dependencies'' since the value of $x$ depends on its current value. These are not well-formed programs and are statically rejected. Instead, (S-LETREC) allows the evaluation of $e_1$, whose value will determine the current value of $x$, to reference $x$ up to and including its previous value, but not its current value. A $\nil$ is inserted at the end of $\sigma(x)$ since the last position is reserved for the current value of a variable (which depends on the current evaluation). The evaluation of $e_2$ is then allowed to reference the current value of $x$, as one would expect in a local binding. Recursively defined streams must delay references to themselves using the $\kw{fby}$ keyword, which also defines the stream's initial value(s). This ensures that streams are causal, depending on strictly prior values. The stream $\letrec{x=0\fby{}x+1}{x}$ is well-formed, because the recursive reference to $x$ in its definition occurs behind a delay. This stream initially produces the value $0$, and then for all subsequent evaluations produces a value one more than its previous value. The variable $x$ is annotated with a type $\tau$ to simplify type checking and this type is likewise ``stepped'' to reflect the behavior of $x$ in subsequent instants. The operator $\tl{}$ and the semantics of ``stepping'' types is described in Section \ref{sec:6}.

    \item \textit{Local Bindings}: This expression type is similar to that of stream recursion, except the evaluation of $e_1$ must not depend on $x$, as this is a non-recursive definition. Otherwise, all other behaviors, including stepping the type annotation of $x$, are the same.
    
    \item \textit{Branching}: Branching in MARVeLus resembles a multiplexer in a digital circuit, in that the conditional selects the value of one of the two streams to emit. That is, the stream ${\sf if~}x_c{\sf ~then~}e_t{\sf ~else~}e_f$ takes the current value of the ``true'' branch $e_t$ or the ``false'' branch $e_f$ depending on the truth value of $x_c$. As before, we assume all three streams have compatible clocks. The two rules (S-IF-T) and (S-IF-F) correspond respectively to the cases when $x_c$ emits \kw{true} or \kw{false}, and thus whether the stream emits the value of $e_t$ (thus emitting $v_t$) or $e_f$ (thus emitting $v_f$). The conditional does not affect the behavior of the branches themselves; the branch not taken will transition as well as reflected in the rules. This prevents the stream that is not selected from delaying indefinitely.
\end{enumerate}

\begin{figure*}
    \flushleft{$\boxed{S;~\sigma\vdash e\xrightarrow{v}e'}$}\\
    \centering
    \vquad
    
{\begin{bprooftree}
\AxiomC{}
\RightLabel{(S-CONST)}
\UnaryInfC{$S;~\sigma\vdash c\xrightarrow{c} c$}
\end{bprooftree}}
\qquad
{\begin{bprooftree}
\AxiomC{}
\RightLabel{(S-VAR)}
\UnaryInfC{$S;~\sigma,x=\hist::v\vdash x\xrightarrow{v} x$}
\end{bprooftree}}

\vquad
{\begin{bprooftree}
    \AxiomC{$S;~\sigma\vdash e_1\xrightarrow{v_1} e_1'$}
    \RightLabel{(S-FBY)}
    \UnaryInfC{$S;~\sigma\vdash e_1 \fby{} e_2 \xrightarrow{v_1} \delay{e_2}$}
\end{bprooftree}}
\qquad
{\begin{bprooftree}
    \AxiomC{$S;~\kw{prev}(\sigma)\vdash e\xrightarrow{v} e'$}
    \RightLabel{(S-DELAY)}
    \UnaryInfC{$S;~\sigma\vdash \delay{e}\xrightarrow{v}\delay{e'}$}
\end{bprooftree}}
\\
\vquad
{\begin{bprooftree}
\AxiomC{$S;~\sigma, x=\hist::\true \vdash e_t \xrightarrow{v_t} e_t'$}
\AxiomC{$S;~\sigma, x=\hist::\true \vdash e_f \xrightarrow{v_f} e_f'$}
\RightLabel{(S-IF-T)}
\BinaryInfC{$S;~\sigma, x=\hist::\true \vdash (\kw{if~} x \kw{~then~} e_t \kw{~else~} e_f) \xrightarrow{v_t} (\kw{if~} x \kw{~then~} e_t' \kw{~else~} e_f')$}
\end{bprooftree}}
\\
\vquad
{\begin{bprooftree}
\AxiomC{$S;~\sigma, x=\hist::\false  \vdash e_t \xrightarrow{v_t} e_t'$}
\AxiomC{$S;~\sigma, x=\hist::\false  \vdash e_f \xrightarrow{v_f} e_f'$}
\RightLabel{(S-IF-F)}
\BinaryInfC{$S;~\sigma, x=\hist::\false \vdash (\kw{if~} x \kw{~then~} e_t \kw{~else~} e_f) \xrightarrow{v_f} (\kw{if~} x \kw{~then~} e_t' \kw{~else~} e_f')$}
\end{bprooftree}}
\\
\vquad

{\begin{bprooftree}
\AxiomC{$S;~\sigma,x=\hist::\nil\vdash e_1\xrightarrow{v}e_1'$}
\AxiomC{$S;~\sigma,x=\hist::v\vdash e_2\xrightarrow{w}e_2'$}
\RightLabel{(S-LETREC)}
\BinaryInfC{$S;~\sigma\vdash \kw{let~rec}_{\hist}~x:\tau=e_1\kw{~in~}e_2\xrightarrow{w}\kw{let~rec}_{\hist::v}~x:\tl{\tau}=e_1'\kw{~in~}e_2'$}
\end{bprooftree}}
\vquad

{\begin{bprooftree}
\AxiomC{$S;~\sigma\vdash e_1\xrightarrow{v}e_1'$}
\AxiomC{$S;~\sigma,x=\hist::v\vdash e_2\xrightarrow{w}e_2'$}
\RightLabel{(S-LET)}
\BinaryInfC{$S;~\sigma\vdash \kw{let}_{\hist}~x:\tau=e_1\kw{~in~}e_2\xrightarrow{w}\kw{let}_{\hist::v}~x:\tl{\tau}=e_1'\kw{~in~}e_2'$}
\end{bprooftree}}
\vquad

{\begin{bprooftree}
    \AxiomC{}
    \RightLabel{(S-APP)}
    \UnaryInfC{$S, f(x)=e;\ \sigma, y=\hist::v\vdash f(y) \xrightarrow{e[x\mapsto v]} f(y)$}
\end{bprooftree}}
\vquad

{\begin{bprooftree}
    \AxiomC{$S;~\sigma \vdash e\xrightarrow{v} e'$}
    \AxiomC{$\neg$ \kw{robot-mode}}
    \RightLabel{(S-MODELS)}
    \BinaryInfC{$S;~\sigma \vdash e \kw{~models~} r \xrightarrow{v} e'\kw{~models~}r$}
\end{bprooftree}}

\caption{MARVeLus Stream Semantics, inspired from existing semantics of Lustre and Esterel \cite{ratel1992definition,parissis1996test,colaco2006mixing,caspi_co-iterative_1998}
    }
    \label{fig:streamsemantics}
\end{figure*}

\subsection{Predicate Semantics}
MARVeLus allows one to define predicates in the style of Linear Temporal Logic \cite{pnueli1977temporal}. Adopting the notation of Manna and Pnueli \cite{manna1992temporal}, we use $\vDash$ to denote satisfaction of a property over a stream and $\VDash$ to denote satisfaction over a single state. In this section, we denote streams $s$ as sequences of values $s_i$ for illustrative purposes. Furthermore, we note that predicates themselves are in fact streams of Boolean values in our system, with each value representing the predicate's satisfaction given the state of the stream at that instant. Since specifications in MARVeLus are made with respect to the initial definitions of streams (i.e. before program execution), predicates written by the user are considered over the entire time period of the program, from the beginning of program execution onwards.

\subsubsection{State Predicates}
For a state predicate $p$ and stream $s$, $(s_i,~s_{i+1},~s_{i+2},...)\vDash p$ if and only if $s_i\VDash p$.

\subsubsection{Trace Predicates}
Trace predicates offer a more holistic view of the stream's behavior, and we define a subset of those found in Linear Temporal Logic.
\begin{enumerate}
    \item Globally ($\alw$): A property $\alw\phi$  requires that $\phi$ be satisfied for all values of a stream beginning at the present instant. That is, $(s_i,~s_{i+1},~s_{i+2},...)\vDash (\alw\phi)_i\iff \forall (j\geq i)\ .\ s_j\VDash \phi_i$. Verification of these predicates in MARVeLus is handled via induction; the left- and right-hand sides of a $\kw{fby}$ construct become the initial condition and inductive step, respectively. Since refinement predicates are checked in the program's initial instant, satisfaction of $\alw\phi$ implies satisfaction of $\phi$ for all instants in time; thus $\alw$ properties are invariant properties.
    
    \item Next ($\nxt$): A property $\nxt\phi$ requires that $\phi$ be satisfied in the next instant. In other words, $(s_i,~s_{i+1},~s_{i+2},...)\vDash \nxt\phi\iff  s_{i+1}\VDash \phi$
    
\end{enumerate}

We additionally introduce $\hd$ as a convenience function that transforms a trace predicate $\phi$ such that for any streams $s$ and $t$, $s\vDash \phi \Rightarrow (s\fby{} t) \vDash \hd{\phi}$. Thus, $\hd{\phi}$ represents a condition that a stream with $s_0$ as its initial value will satisfy, regardless of subsequent values, assuming that $s$ satisfies $\phi$. This is useful for deriving properties about a stream constructed using $s\fby{}t$, if properties are known about $s$ and $t$ individually. The transformation is defined recursively for state predicates $p$ and trace predicates $\phi,\psi$ in Fig. \ref{fig:hdsemantics}. The transformation for state predicates and $\alw$ are unsurprising; state predicates are only ever concerned with the immediate value, and $\alw\phi$ follows from the equivalence $\alw\phi \equiv \phi \land \nxt\alw\phi$. No additional information can be learned from $\nxt\phi$ as its satisfaction tells us nothing about the immediate value of the stream. We also introduce the $\kw{impl}()$ operator to condition refinement predicates at each instant on a boolean stream, without introducing disjunctions of trace predicates, to maintain the ability to split predicates.
\begin{lemma}[Correctness of \kw{hd}]
For any stream $\sigma=(\sigma_0,\  \sigma_1,\ldots)$, if $\sigma \vDash \psi$, then for any other stream $\tau=(\tau_0,\ \tau_1, \ldots)$, $(\sigma_0,\ \tau_1,\ \tau_2...\tau_n...)\vDash\kw{hd}(\psi)$
\label{lemma:hdsoundness}
\end{lemma}
\begin{proof}
    By induction on the definition of $\kw{hd}(\psi)$\\
    \textbf{Case} $p$ where $p$ is a state predicate. 
    Assume $\sigma \vDash p$. Then $\sigma_0 \VDash p$. Therefore, $(\sigma_0,\ \tau_1,\ \tau_2...\tau_n...)\vDash \kw{hd}(p)\equiv p$ for any $\tau$. \\
    \textbf{Case} $\psi_1\land\psi_2$. 
    Assume $\sigma \vDash \psi_1\land\psi_2$. Then $\sigma\vDash \psi_1$ and $\sigma\vDash \psi_2$. By IH, $(\sigma_0,\ \tau'_1,\ \tau'_2...\tau'_n...)\vDash \kw{hd}(\psi_1)$ for any $\tau'$, and $(\sigma_0,\ \tau'_1,\ \tau'_2...\tau'_n...)\vDash \kw{hd}(\psi_2)$ for any $\tau'$. Therefore, $(\sigma_0,\ \tau'_1,\ \tau'_2...\tau'_n...)\vDash \kw{hd}(\psi_1\land\psi_2)$.\\
    \textbf{Case} $\alw\psi$. 
    Assume $\sigma \vDash \alw\psi$. By induction hypothesis (IH), $(\sigma_0,\ \tau'_1,\ \tau'_2...\tau'_n...)\vDash \kw{hd}(\psi)$ for any $\tau'$. Therefore,  $(\sigma_0,\ \tau_1,\ \tau_2...\tau_n...)\vDash \kw{hd}(\alw\psi)\equiv \kw{hd}(\psi)$ for any $\tau$.\\
    \textbf{Case} $\nxt\psi$. 
    Trivial. $(\sigma_0,\ \tau_1,\ \tau_2...\tau_n...)\vDash \kw{hd}(\nxt\psi)\equiv \top$ for any $\sigma$ and $\tau$.
\end{proof}
\begin{figure}
    \centering
    \begin{align*}
        \kw{hd}(p) & \triangleq p &
        \kw{impl}(q, p) & \triangleq \lnot q\lor p \\
        \kw{hd}(\phi\land \psi) & \triangleq \kw{hd}(\phi) \land \kw{hd}(\psi) &
        \kw{impl}(q, \phi\land\psi) & \triangleq \kw{impl}(q, \phi)\land \kw{impl}(q, \psi)\\
        \kw{hd}(\alw\phi) & \triangleq \kw{hd}(\phi) &
        \kw{impl}(q, \alw\phi) & \triangleq \alw(\kw{impl}(q, \phi)) \\
        \kw{hd}(\nxt\phi) & \triangleq \top&
        \kw{impl}(q, \nxt\phi) & \triangleq \nxt(\kw{impl}(q, \phi))
    \end{align*}
    \caption{Inductive definitions of the \kw{hd} and \kw{impl} operators. The \kw{impl} operator is needed because the syntax only allows the $\lor$ operator (and thus implications) at the level of state formulas.}
    \label{fig:hdsemantics}
\end{figure}

\section{The MARVeLus Type System}
\label{sec:6}
Formal verification of MARVeLus programs relies upon type checking using specifications supplied as type refinements. Our primary contribution to the type system is the adaptation and implementation of refinement type specifications for streams. Unique to our refinement type system is the introduction of temporal operators in refinement predicates. 

\begin{figure}
\flushleft{$\boxed{G;~H\vdash e:\tau}$}\\
\vquad

    \centering
    {\begin{bprooftree}
    \AxiomC{}
    \RightLabel{(T-CONST)}
    \UnaryInfC{$G;~H\vdash c:\refinement{w}{b}{\alw(w=c)}$}
    \end{bprooftree}}
    \qquad
    {\begin{bprooftree}
    \AxiomC{$G;~H(x)=\refinement{w}{b}{\psi}$}
    \RightLabel{(T-VAR)}
    \UnaryInfC{$G;~H\vdash x:\refinement{w}{b}{\psi \land \alw(w=x)}$}
    \end{bprooftree}}
    \vquad

    {\begin{bprooftree}
        \AxiomC{$G;~H\vdash e_1:\refinement{w}{b}{\hd{\psi_1}}$}
        \AxiomC{$G;~H\vdash e_2:\refinement{w}{b}{\psi_2}$}
        \RightLabel{(T-FBY)}
        \BinaryInfC{$G;~H\vdash e_1\fby{} e_2:\refinement{w}{b}{\hd{\psi_1} \land \nxt\psi_2}$}
    \end{bprooftree}}
    \vquad
    
    {\begin{bprooftree}
        \AxiomC{$G;~\prev{H}\vdash e:\tau$}
        \RightLabel{(T-DELAY)}
        \UnaryInfC{$G;~H\vdash \delay{e}:\tau$}
    \end{bprooftree}}
    \vquad
    
    {\begin{bprooftree}
    \AxiomC{$G;~H\vdash \tau_1$}
    \AxiomC{$G;~H \vdash e_1:\tau_1$}
    \AxiomC{$G;~H,x:\tau_1 \vdash e_2:\tau_2$}
    \RightLabel{(T-LET)}
    \TrinaryInfC{$G;~H \vdash {\sf let}_h~{x:\tau_1=e_1}{\sf~in~}{e_2}:\tau_2$}
    \end{bprooftree}}
    \vquad
    
    {\begin{bprooftree}
    \AxiomC{$G;~H\vdash \tau_1$}
    \AxiomC{$G;~H,x:\tau_1\vdash e_1:\tau_1$}
    \AxiomC{$G;~H,x:\tau_1 \vdash e_2:\tau_2$}
    \RightLabel{(T-LETREC)}
    \TrinaryInfC{$G;~H\vdash {\sf let}_h{\sf~rec~}x:\tau_1=e_1{\sf ~in~} e_2:\tau_2$}
    \end{bprooftree}}\\

    \vquad

    {\begin{bprooftree}
        \AxiomC{$G;~H\vdash x_c:\kw{bool}$}
        \AxiomC{$G;~H\vdash e_t:\refinement{w}{b}{\kw{impl}(x_c, \psi)}$}
        \noLine
        \UnaryInfC{$G;~H\vdash e_f:\refinement{w}{b}{\kw{impl}(\lnot x_c, \psi)}$}
        \RightLabel{(T-IF)}
        \BinaryInfC{$G;~H\vdash \ifthenelse{x_c}{e_t}{e_f}:\refinement{w}{b}{\psi}$}
    \end{bprooftree}}

    \vquad

    {\begin{bprooftree}
        \AxiomC{$G,f:(x:\refinement{w_1}{b_1}{\phi_1} \to \refinement{w_2}{b_2}{\phi_2});~H\vdash y:\refinement{w_1}{b_1}{\alw\phi_1}$}
        \RightLabel{(T-APP)}
        \UnaryInfC{$G,f:(x:\refinement{w_1}{b_1}{\phi_1} \to \refinement{w_2}{b_2}{\phi_2});~H\vdash f~y:\refinement{w_2}{b_2}{\alw\phi_2[x\mapsto y]}$}
    \end{bprooftree}}
    
\vquad

    {\begin{bprooftree}
        \AxiomC{$G;~H\vdash e:\tau$}
        \RightLabel{(T-MODELS)}
        \UnaryInfC{$G;~H\vdash e\kw{~models~}r:\tau$}
    \end{bprooftree}}
    \qquad
    {\begin{bprooftree}
    \AxiomC{}
    \RightLabel{(T-NIL)}
    \UnaryInfC{$G;~H\vdash \nil:\tau$}
\end{bprooftree}}

\vquad

\begin{bprooftree}
    \AxiomC{$G;~H\vdash e:\tau$}
    \AxiomC{$G;~H\vdash \tau \preceq \tau'$}
    \AxiomC{$G;~H\vdash\tau'$}
    \RightLabel{(T-SUB)}
    \TrinaryInfC{$G;~H\vdash e:\tau'$}
\end{bprooftree}
    \caption{We define a set of stream typing rules incorporating the aforementioned temporal semantics. We use a subtyping rule (T-SUB) similar to that of Liquid Haskell \cite{jhala2020refinement,borkowski_mechanizing_2022}.}
    \label{fig:typingrules}
\end{figure}

\begin{figure}
\flushleft{$\boxed{G;~H\vdash \tau_1\preceq\tau_2}$}\\
\vquad
\centering
    \begin{bprooftree}
        \AxiomC{$G;~H \vdash \forall x_1:b~.~\phi_1 \Rightarrow \phi_2[x_2 \mapsto x_1]$}
        \RightLabel{(SUB-BASE)}
        \UnaryInfC{$G;~H\vdash \refinement{x_1}{b}{\phi_1}\preceq \refinement{x_2}{b}{\phi_2}$}
    \end{bprooftree}
    
    \vquad
    
    \caption{Subtyping rule inspired by Liquid Haskell \cite{jhala2020refinement,borkowski_mechanizing_2022}, which we implement in MARVeLus. The arrow $\Rightarrow$ symbolizes logical implication.}
    \label{fig:subtyping}
\end{figure}

\begin{figure}
    \flushleft{$\boxed{G;~H\vdash c}$}\\
    \vquad
    \centering
    \begin{bprooftree}
        \AxiomC{$\kw{SMTValid}(c)$}
        \RightLabel{(ENT-EMP)}
        \UnaryInfC{$G;~\varnothing\vdash c$}
    \end{bprooftree}
    \qquad
    \begin{bprooftree}
        \AxiomC{$G;~H\vdash \forall x:b~.~\phi[v \mapsto x] \Rightarrow c$}
        \RightLabel{(ENT-EXT)}
        \UnaryInfC{$G;~H, x:\refinement{v}{b}{\phi}\vdash c$}
    \end{bprooftree}
    
    \caption{We use similar entailment rules as those found in \cite{jhala2020refinement} for discharging verification conditions, where \kw{SMTValid}() is a call to the SMT solver to check validity of a predicate.}
    \label{fig:entailment}
\end{figure}

Fig. \ref{fig:typingrules} gives the set of typing rules, adopting the proposed stream syntax. 

\subsection{Type Semantics}
A stream's type refinement specifies the stream's behavior from the current instant onwards. For example, a stream inhabiting the type $\refinement{w}{\inttype}{\alw(w\geq 0)}$ is non-negative at the current instant, and will continue to be non-negative in all subsequent time instants. We follow the LTL convention in that predicates are checked for validity in the time instant in which they occur. Thus, a stream's type describes the stream's immediate value along with the values it will produce in the future. We use the same judgement as Zélus \cite{benveniste2011hybrid}--$G,H \vdash e:\tau$--to express that, in function context $G$ and stream context $H$, an expression $e$ has type $\tau$. It has proven to be more manageable to separate imported functions from the base language into a distinct environment $G$.

\subsection{Environments, Types, and their Operations}
Aside from the separation of function and stream environments, the two environments are defined in the standard way (Fig. \ref{fig:h-ops}). By virtue of our streams being constructed from initial values and subsequent steps (through \kw{fby}), it is convenient to define operators $\hd{}$ and $\tl{}$ to extract the ``immediate'' and ``suffix'' of a type, respectively. For some type $\tau$ that can be written as $\refinement{w}{b}{\hd{\phi_1}\land \nxt\phi_2}$, an expression of type $\hd{\tau}$ is only required to satisfy the refinement predicate of $\tau$ in the current step. Hence, it only needs to satisfy $\hd{\phi_1}$. Similarly, an expression of type $\tl{\tau}$ must satisfy, immediately, the refinement predicate that an expression of type $\tau$ will satisfy in the next time instant. The $\tl{}$ operator is particularly useful for computing the types of rewritten expressions of the form $e\xrightarrow{v}e'$, and so we extend it to operate on the entire environment ($\tl{H}$). Analogously to the operational semantics, we must also capture a notion of ``history'' in the typing context. However, the result is not as exact. If $\sigma$ represents a variable's past behavior, then $H$ represents its future. As a variable is updated, its immediate future behavior specializes into a single value, which in the next time instant, gets added to $\sigma$. This process ``consumes'' the part of the refinement predicate that pertains to the current instant in preparation for determining the variable's future behavior. Unlike the case for $\sigma$, reversing this process in $H$ is lossy, and the best we can accomplish is an approximation of the original type. Thus, if we started with a type $\refinement{w}{b}{\hd{\phi_1}\land\nxt \phi_2}$ and take its suffix $\tl{\refinement{w}{b}{\hd{\phi_1}\land\nxt \phi_2}}=\refinement{w}{b}{\phi_2}$, we may not be able to recover $\hd{\phi_1}$, at least not syntactically. However, we at least know that  $\prev{\tl{\refinement{w}{b}{\hd{\phi_1}\land\nxt \phi_2}}}$ is a subtype of $\refinement{w}{b}{\nxt\phi_2}$. Using this, and potentially additional information brought in through invariants, we can approximate the type environment from the previous time instant, which may suffice even if we cannot recover the exact environment.

\begin{figure}
    \centering
    \begin{align*}
        H ::= & \varnothing~|~H,x:\tau\\
        G ::=&\varnothing~|~G,f(x):x:\tau_1 \to \tau_2\\
        \hd{\refinement{w}{b}{\hd{\phi_1}\land \nxt \phi_2}} \triangleq & \refinement{w}{b}{\hd{\phi_1}}\\
        \tl{\refinement{w}{b}{\hd{\phi_1}\land \nxt \phi_2}} \triangleq & \refinement{w}{b}{\phi_2}\\
        \tl{H} \triangleq & \{\tl{H(x)}~|~x\in\dom{H}\}\\
        \refinement{w}{b}{\hd{\phi_1}\land \nxt \phi_2} &\preceq \prev{\refinement{w}{b}{\phi_2}} \preceq \refinement{w}{b}{\nxt \phi_2}\\
        \prev{H} \triangleq & \{\prev{H(x)}~|~x\in\dom{H}\}
    \end{align*}
    \caption{Typing Environments, and the operations on types and environments.}
    \label{fig:h-ops}
\end{figure}

\subsection{Summary of Typing Rules}
\begin{enumerate}
\item \textit{Constants}: (T-CONST) encapsulates the behavior of refinement type synthesis for constants \cite{jhala2020refinement}, assigning a specific singleton type to constants. Since constants do not change, we can strengthen the refinement predicate with $\square$ to require that the constant $c$ is typed as a stream that is always equal to $c$. 

\item \textit{Variables}: (T-VAR) performs selfification on variables \cite{jhala2020refinement,ou_dynamic_2004}, which strengthens the type by explicitly requiring that the refinement variable $v$ is always equal to the term variable $x$. This is achieved by the addition of the $\square(v=x)$ constraint in the refinement. 

\item \textit{Stream Compositions}: (T-FBY-0) and (T-FBY-1) generate the verification conditions for typing the composition of two streams $e_1$ and $e_2$, with associated predicates $\hd{\psi_1}$ and $\psi_2$, respectively. Since the resultant stream must not change its base type, we require that $e_1$ and $e_2$ type to the same base type. The \kw{fby} construct allows one to combine values of both streams $e_1$ and $e_2$; naturally, the resultant stream inherits some properties of both. Crucially, the stream only takes the first value of $e_1$, after which it replicates the behavior of the stream $e_2$ but delayed by one cycle. Consequently, given $e_1$ satisfying property $\psi_1$ for at least the first cycle and $e_2$ satisfying $\psi_2$, the stream $e_1$ \kw{fby} $e_2$ should satisfy both $\psi_1$ for the first cycle, and also $\nxt\psi_2$. The subscript denotes the value ``buffered'' by \kw{fby}, which is empty in the first cycle (T-FBY-0), and may be occupied by a value in subsequent cycles (T-FBY-1). Since the first stream is only considered in the first cycle and is disregarded in determining the resultant type in (T-FBY-1), it is left unrefined. Instead, a stream's behavior with $\fby{u}$ where $u\neq \nil$ is determined by the $u$ and $e_2$.

\item \textit{Local Bindings}: We define recursive (T-REC) and non-recursive (T-LET) local bindings as usual, noting that (T-REC) is crucial for proving inductive invariants when paired with (T-FBY). We note that T-LET requires type annotations be well-formed ($G;~H\vdash \tau_1$), that is, the refinement predicate is a Boolean stream that only references in-scope variables. This is a direct lifting of the non-stream well-formedness judgement for other refinement type systems \cite{jhala2020refinement,rondon2008liquid}. Suppose we want to verify that the program \texttt{let rec x:\{v:int | v >= 0\} = 0 fby x + 1 in x} types to $\{v:\kw{int}~|~v\geq 0\}$. Since $x$ is recursively defined, verifying that all values of $x$ are non-negative involves first verifying that the initial condition 0 and then the inductive definition \texttt{x+1}, with the assumption that $\square(x\geq 0)$ held in the previous instant. At first glance, the assumption $\square(x\geq 0)$ appears circular, essentially stating that the stream we want to check already satisfies the predicate we want to verify. However, we note that at each instant beyond the first time step, $x$ represents the stream in the current time step, whereas $x+1$ is the new definition that will take effect in the next time step. Therefore, the verification task shifts to proving that, assuming the stream $x$ already satisfies $\square(x\geq 0)$, the predicate $\square((x+1)\geq 0)$ is valid.

\item \textit{Branching Expressions}: As with other refinement type systems \cite{jhala2020refinement}, the type to be checked is passed to each branch and further refined so that it only needs to be satisfied when the conditional of that branch is satisfied. This control-flow awareness prevents the type system from needlessly rejecting terms that the program will have never reached. This notion is lifted to stream types through the $\kw{impl}()$ operator, which inserts an implication statement at the state-predicate level, achieving the same effect as in non-stream refinement type systems with a similar rule. 

\item \textit{Function Application}: Since functions are imported from a non-stream language and applied point-wise to the values of stream variables, they are considered constant and thus their types are unchanging over time. As a result, the argument variable should type to the function's argument type at all times. By extension, the return type of the function application should be enforced for all time. This is encapsulated by the (T-APP) rule, which augments the original (non-stream) function type with the LTL $\alw$ operator to require that these properties hold over all time. 

\item \textit{Robot-specific commands}:
Due to the unpredictability of real-world physical parameters, we exclude \kw{robot\_get} and \kw{robot\_store} from refinement type-checking, and treat them as functions that simply return the appropriate un-refined term for the base Zélus type checker. Instead, we require that each call to \kw{robot\_get} be accompanied by \kw{models}, so that each variable intended to derive values from the real world has a deterministic verifiable model. Indeed, the user is responsible for ensuring that the model does in fact reflect expected real-world conditions, which are impossible to predict at compile-time. Since the \kw{models} keyword is intended to be semantically transparent, it directly inherits the type of its inner expression. Although it is outside the control of MARVeLus, it is assumed that the robot-specific command following \kw{models} has the same type. 

\item \textit{Typing \nil}:
Since \nil~ can be pushed onto the history of any variable, we allow it to type to anything. While this may appear disconcerting, the static initialization and causality analyses in Zélus \cite{benveniste2011hybrid} ensure that well-formed programs do not allow \nil to be executed. That is, a variable that does not have immediate self-dependencies or have circular dependencies with other variables will never access a variable whose history ends in \nil. 

\item \textit{Subtyping}:
We use the subtyping rule from Liquid Haskell \cite{jhala2020refinement,rondon2008liquid}, which stipulates that an expression which already types to $\tau$ may also type to $\tau'$ if $\tau'$ is a subtype of $\tau$. This allows refinement types to be checked without needing an exact match to the user-supplied type. We use the standard well-formedness rules surrouding type refinements to ensure $\tau'$, which may be user-supplied, has a well-typed predicate and does not contain unbound variables. We also adopt the subtyping judgement $G;H\vdash \tau_1\preceq \tau_2$ with rule (SUB-BASE) (Fig. \ref{fig:subtyping} to subtype refined base types in a similar style as Liquid Haskell. A function subtyping rule is not necessary as MARVeLus streams do not support higher-order functions. We use a similar entailment judgement $G;H\vdash c$ (Fig. \ref{fig:entailment}) which provides the interface to the SMT solver by constructing the verification condition to be checked for validity given the subtyping implication and refinement predicates bound to variables in the typing environment. 

\end{enumerate}

\subsection{Type Safety}
We prove type safety roughly following the classical syntax-guided manner of Wright and Felleisen \cite{wright1994syntactic}. A key difference arises in our type system from the decoupling of program terms from the values they produce in the synchronous semantics. Type safety must apply not only to the values emitted, but also to the subsequent evaluations of the stream term to ensure that the program remains safe in the future. We also note that existing work by Boulmé and Hamon \cite{goos_certifying_2001} has explored synchronous properties of languages such as Lustre, some of which resemble progress in conventional languages. 
We begin by defining convenience operators $\kw{basetype}$ and $\kw{now}$, which ``un-refine'' a refined type and retrieve the most recent entry of a variable's history, respectively:
\begin{align*}
    \kw{basetype}(\refinement{w}{b}{\phi})=b &&\kw{now}(\hist::v)=v
\end{align*}
We also define conditions under which the function and stream term environments correspond with their type environment counterparts. Namely, we want to enforce some notion of well-typedness on the environments:

\begin{defn}[Function Environment Correspondence]
    $G:\kw{name}\to T$ \textbf{corresponds with} $S:\kw{name}\to(x:\tau_1\to\tau_2)$ when the following conditions are met:
    \begin{itemize}
        \item $\dom{G} \subseteq \dom{S}$
        \item $\forall f \in \dom{G}$, $G;H \vdash S(f):G(x)$
    \end{itemize}
    \label{def:correspfunction}
\end{defn}

\begin{defn}[Stream Environment Correspondence]
    $H:\kw{name}\to T$ \textbf{corresponds with} $\sigma:\kw{name}\to h$ when the following conditions are met:
    \begin{itemize}
        \item $\dom{H} \subseteq \dom{\sigma}$
        
        \item $\forall x \in \dom{H}$:
        \begin{itemize}
            \item $G; H \vdash \sigma(x): \kw{basetype}(H(x))$ and 
            \item $G;H \vdash \kw{now}(\sigma(x)):\hd{H(x)}$
        \end{itemize}
        \item $\prev{\sigma}$ corresponds with any $\prev{H}$
    \end{itemize}
    \label{def:correspstream}
\end{defn}

Definition \ref{def:correspfunction} is standard. We note the absence of $H$ in its typing judgement as functions are judged in their original non-stream contexts using the standard typing judgement. Definition \ref{def:correspstream} involves streams and thus requires different treatment. Stream types determine a variable's current and expected future trajectory, while stream histories record a variable's past and current trajectory. They intersect at the present moment, where both the history and type of a variable refer to the variable's current state. Thus, at any given moment, we can only make a precise, refinement typing judgement about the variable's current value. Since histories are initialized to be empty and are extended only by let-bindings (recursive or otherwise), we can ensure that let-bindings maintain correspondence between the type and term environments for all time. This is proven explicitly in the (T-LETREC) case in our type preservation proof, and assumed elsewhere. Because the $\prev{}$ operator is unable to exactly recover the original refinement of a type, we simply require correspondence of $\prev{\sigma}$ with any $\prev{H}$, noting that for any $H$ which evolves into $H'$, $\prev{H'}$ will necessarily produce subtypes of variables in $H$. We can however ascertain that a variable corresponds with its base type at any point in time.

\begin{theorem}[Type Safety]
    If $~H$ corresponds with $\sigma$ and $G$ corresponds with $S$, and if $~G;H\vdash e:\tau$, then $\exists e'$ such that $S;\sigma\vdash e\xrightarrow{v}e'$ where $v$ is a value, $G,H \vdash v:\hd{\tau}$, and $G,\tl{H}\vdash e':\tl{\tau}$
    \label{thm:typesafety}
\end{theorem}

\begin{proof}
    By proving progress and preservation on well-typed terms.
\end{proof}

\subsubsection{Progress}
The standard notion of progress stipulates that a well-typed program term is either a value or can undergo reduction via the semantic rules. Since all MARVeLus streams are designed to emit values in perpetuity, including constant streams, we must show that there is always a rule that allows the program to take another semantic step, and thus emit another value. 
\begin{lemma}[Progress]
    If $~H$ corresponds with $\sigma$ and $G$ corresponds with $S$, if $~G;H\vdash e:\tau$, then $\exists e'$ such that $S;\sigma\vdash e\xrightarrow{v}e'$ where $v$ is a value.
    \label{lemma:progress}
\end{lemma}

\begin{proof}
        By structural induction on the typing derivation of $e$.\\
        \textbf{Case} T-CONST.
        \begin{itemize}
            \item $G;H\vdash c:\refinement{w}{b}{\alw(w=c}$
            \item By S-CONST, $S;\sigma\vdash c\xrightarrow{c}c$
            \item By definition, $c$ is a value.
        \end{itemize}
        \textbf{Case} T-VAR.
        \begin{itemize}
            \item $G;H\vdash x:\refinement{w}{b}{\psi\land \alw(w=x)}$ for some refinement predicate $\psi$
            \item Because $x$ is well-typed, $x\in \dom{H}$.
            \item By assumption, $x\in\dom{\sigma}$, and so $\sigma(x)=h::v$ for some history $h::v$ where $G;H\vdash v:\hd{H(x)}$ 
            \item Therefore $S;\sigma,x=h::v\vdash x\xrightarrow{v}x$
            \item By definition of histories $v$ is a value.
        \end{itemize}
        \textbf{Case} T-FBY
        \begin{itemize}
            \item $G;H\vdash e_1\fby{} e_2:\refinement{w}{b}{\hd{\psi_1}\land \nxt \hd{\psi_2}}$
            \item By assumption, $G; H\vdash e_1:\refinement{w}{b}{\psi_1}$ and $G; H\vdash e_2:\refinement{w}{b}{\psi_2}$
            \item By IH, $S;\sigma \vdash e_1 \xrightarrow{v_1}e_1'$ and $S;\sigma \vdash e_2 \xrightarrow{v_2}e_2'$ where $v_1$ and $v_2$ are values.
            \item By S-FBY, $S;\sigma \vdash e_1 \fby{} e_2 \xrightarrow{v_1}\delay{e_2}$
        \end{itemize}
        \textbf{Case} T-DELAY
        \begin{itemize}
            \item $G;H \vdash \delay{e_1}:\tau$
            \item By assumption, $G;\prev{H}\vdash e:\tau$
            \item Since $H$ and $\sigma$ correspond, $\prev{H}$ and $\prev{\sigma}$ also correspond.
            \item By IH, $S; \prev{\sigma}\vdash e\xrightarrow{v}e'$ where $v$ is a value.
            \item By S-DELAY, $S,\sigma \vdash \delay{e}\xrightarrow{v}\delay{e'}$ 
        \end{itemize}
        \textbf{Case} T-IF
        \begin{itemize}
            \item $G;H \vdash \ifthenelse{x_c}{e_t}{e_f}:\refinement{w}{b}{\psi}$
            \item By assumption:
            \begin{itemize}
                \item $x_c\in \dom{H}$ and thus $x_c\in\dom{\sigma}$ and $\sigma(x)=h::v_c$ for some $h$ and $v_c$ where $G;H\vdash v_c:\hd{H(x)}$
                \item $G;H \vdash x_c:\kw{bool}$
                \item $G;H \vdash e_t:\refinement{w}{b}{\kw{impl}(x_c,\psi)}$
                \item $G;H \vdash e_f:\refinement{w}{b}{\kw{impl}(\neg x_c,\psi)}$
            \end{itemize}
            \item By IH:
            \begin{itemize}
                \item $S; \sigma,x=\hist::v_c \vdash x_c\xrightarrow{v_c}x_c$ where $v_c$ is a value
                \item $S; \sigma,x=\hist::v_c \vdash e_t \xrightarrow{v_t} e_t'$ where $v_t$ is a value
                \item $S; \sigma,x=\hist::v_c \vdash e_f \xrightarrow{v_f} e_f'$ where $v_f$ is a value
            \end{itemize}
            \item Since $v_c$ is a value and a boolean, it can either be \kw{true} or \kw{false}.
            \item \textbf{Case} $v_c=\kw{true}$: By S-IF-T, $S; \sigma,x=\hist::v_c \vdash \ifthenelse{x_c}{e_t}{e_f} \xrightarrow{v_t}\ifthenelse{x_c}{e_t'}{e_f'}$
            \item \textbf{Case} $v_c=\kw{false}$: Symmetric on S-IF-F
        \end{itemize}
        \textbf{Case} T-LETREC
        \begin{itemize}
            \item $G;H \vdash \kw{let}_h\kw{~rec~} x:\tau_1 = e_1 \kw{~in~} e_2:\tau_2$
            \item By assumption:
            \begin{itemize}
                \item $\tau_1$ is well-formed in $G;H$
                \item $G;H,x:\tau_1\vdash e_1:\tau_1$
                \item $G;H,x:\tau_1\vdash e_2:\tau_2$
            \end{itemize}
            \item Since $\sigma$ corresponds with $H$, $x\in\dom{\sigma}$ and thus $\sigma(x)=h::v$ where $G;H\vdash v:\hd{H(x)}$
            \item By IH:
            \begin{itemize}
                \item $S;\sigma,x=h::\nil\vdash e_1\xrightarrow{v}e_1'$
                \item $S;\sigma,x=h::v\vdash e_2\xrightarrow{v_2} e_2'$ where $v_2$ is a value.
            \end{itemize}
            \item By S-LETREC, $S,\sigma \vdash \kw{let}_h\kw{~rec~} x:\tau_1 = e_1 \kw{~in~} e_2\xrightarrow{v_2}\kw{let}_{h::v}\kw{~rec~} x:\tl{\tau_1} = e_1' \kw{~in~} e_2'$
        \end{itemize}
        \textbf{Case} T-APP
        \begin{itemize}
            \item $G,f:x:\tau_1\to \refinement{w_2}{b_2}{\psi_2}; H \vdash f~y:\refinement{w_2}{b_2}{\alw \psi_2[x\mapsto y]}$
            \item By assumption, $y\in\dom{H}$
            \item Because $G$ corresponds with $S$, $f\in \dom{S}$ and is in particular some $f(x)=e$ where $G\vdash f:x:\tau_1\to \refinement{w_2}{b_2}{\psi_2}$
            \item Because $H$ corresponds with $\sigma$, $y\in \dom{\sigma}$ and $\sigma(y)=h::v$ where $G;H\vdash v:\hd{H(y)}$
            \item By S-APP, $S,f(x)=e; \sigma,y=h::v\vdash f(y)\xrightarrow{e[x\mapsto v]}f(y)$
        \end{itemize}
        \textbf{Case} T-MODELS
        \begin{itemize}
            \item $G;H\vdash \models{e}{r}:\tau$
            \item Assume $\neg\kw{robot~mode}$
            \item By assumption, $G;H\vdash e:\tau$
            \item By IH, $S;\sigma\vdash e\xrightarrow{v}e'$ and $v$ is a value
            \item By S-MODELS, $S;\sigma\vdash \models{e}{r} \xrightarrow{v} \models{e'}{r}$
        \end{itemize}
        \textbf{Case} T-SUB
        \begin{itemize}
            \item $G;H\vdash e:\tau'$
            \item Based on the conditions laid out in \cite{borkowski_mechanizing_2022}, we are able to ``collapse" repeated invocations of T-SUB into a single subtyping derivation.
            \item Then, the premise $G;H\vdash e:\tau$ must have been derived through a non-subtyping rule.
            \item By IH, $S;\sigma \vdash e\xrightarrow{v}e'$
        \end{itemize}
\end{proof}

\subsubsection{Preservation}
When a stream expression $e$ is evaluated, it not only produces a value $v$ representing its immediate state at this instant, but also becomes $e'$, representing the stream to be evaluated in the next instant. Over the course of several evaluations, the stream that was once $e$ becomes ``realized'' as a series of emitted values $v_0,\ v_1,\ v_2...$ and stream definitions $e', e'', e'''...$. Combined with the fact that refinement predicates are considered in the time instant in which a stream is evaluated, this means that a predicate that accurately describes a stream in one instant may no longer be satisfied over subsequent evaluations. Consider the stream $s$ whose next values will be $[-1, 0, 1, 2, 3...]$. The predicate $\phi=(s < 0) \land \nxt \alw (s \geq 0)$ accurately describes $s$ in the current time instant. However, after $s$ emits a value and becomes $s'$, whose next values are $[0, 1, 2, 3...]$, $\phi$ no longer accurately describes $s'$, as the value it emits next is not negative. The state predicate $s < 0$ within $\phi$ should only apply to the first state of $s$. An implication that arises is that type refinements also need to evolve over time to accurately describe stream behavior, a phenomenon that is not normally seen in refinement type systems. This somewhat complicates proofs of type safety, as stream expressions that step do not necessarily continue to inhabit their original types. 

However, we note that a stream $e'$ produced by evaluating $e$ is a ``suffix'' of $e$; in other words, the values emitted by $e'$ are identical to those emitted by $e$ after the first value. Likewise, if $e$ satisfies a predicate $\phi$, then $e'$ will satisfy its ``suffix''. This allows us to relate the types that $e$ and $e'$ inhabit and derive a modified notion of type safety, even if their types do not exactly match. If a predicate $\phi$ can be written as $\hd{\phi_1} \land \nxt \phi_2$ for some secondary predicates $\phi_1$ and $\phi_2$, then its suffix is $\phi_2$. Thus, if $e$ satisfies $\phi$, and in the current context $e\xrightarrow{v}e'$, then $v$ satisfies $\hd{\phi_1}$ and $e'$ satisfies $\phi_2$. From this, we can derive types such that if $G;~H\vdash e:\refinement{w}{b}{\hd{\phi_1} \land \nxt \phi_2}$, then $G;~H \vdash v:\refinement{w}{b}{\hd{\phi_1}}$ and $G;~\tl{H}\vdash e':\refinement{w}{b}{\phi_2}$ where $\tl{H}$ is the immediate successor of $H$:
$$\dom{\tl{H}}=\dom{H},~ \forall ~x\in \dom{\tl{H}}~.~(\tl{H}(x))=\tl{H(x)}$$
where for a refined type $\tau=\refinement{w}{b}{\hd{\phi_1}\land \nxt \phi_2}$, $\tl{\tau}=\refinement{w}{b}{\phi_2}$. We prove that $\phi_1$ and $\phi_2$ exist for any $\phi$ constructed in our specification grammar:

\noindent\begin{lemma}[Temporal Predicate Expansion]
\label{lemma:expansionlemma}
Let $\psi$ be a temporal predicate. Then there exist $\psi_1,\ \psi_2$ such that $\psi \equiv \hd{\psi_1}\land \nxt \psi_2$
\end{lemma}
\begin{proof}
By induction on the specification grammar (Fig. \ref{fig:typesyntax})
\\
\textbf{Case} $p$: $\Psi \equiv p$ (where $p$ is a state predicate):
then $\Psi_1\equiv p$ and $\Psi_2\equiv \top$ satisfy the desired condition.

\noindent\textbf{Case} $\land$: $\Psi \equiv \psi_l\ \land\ \psi_r$ (where $\land$ is a binary logical connective):
by IH, $\exists \psi_{l,1},\ \psi_{l,2}$ s.t. $\psi_l \equiv \hd{\psi_{l,1}} \land \nxt \psi_{l,2}$ and likewise $\exists \psi_{r,1},\ \psi_{r,2}$ s.t. $\psi_r \equiv \hd{\psi_{r,1}} \land \nxt \psi_{r,2}$. 
Then $\Psi_1\equiv \hd{\psi_{l,1}\land \psi_{r,1}}$ and $\Psi_2 \equiv (\psi_{l,2} \land \psi_{r,2})$ satisfy the desired condition.

\noindent\textbf{Case} $\nxt$: $\Psi \equiv \nxt \psi$, therefore $\Psi_1\equiv \top$ and $\Psi_2 \equiv \psi$ satisfy the desired condition.

\noindent\textbf{Case} $\alw$: We have $\Psi \equiv \alw \psi$.
By IH, $\exists \psi_1,\ \psi_2$ s.t. $\psi\equiv \hd{\psi_1}\land \nxt \psi_2$. Note that $\alw\psi\equiv \psi\land\nxt\alw\psi$. 
Therefore $ \alw\psi\equiv (\hd{\psi_1}\land \nxt \psi_2) \land\nxt\alw\psi \equiv \hd{\psi_1}\land \nxt(\psi_2 \land \alw\psi)$.
Then $\Psi_1\equiv \psi_1$ and $\Psi_2\equiv \psi_2 \land \alw\psi$ satisfy the desired condition.

\end{proof}

\begin{lemma}[Type Preservation]
\label{lemma:preservation}
Let $e$ be a synchronous stream expression. If\ $H$ corresponds with $\sigma$ and $G$ corresponds with $S$, then if $G,H\vdash e:\tau$ and $S;~\sigma \vdash e\xrightarrow{v}e'$, then $G,H \vdash v:\hd{\tau}$ and $G,\tl{H}\vdash e':\tl{\tau}$
\end{lemma}
\begin{proof}
By structural induction on the operational semantics derivation. We prove by cases on the last rule used. 
\\
\noindent\textbf{Case} S-LETREC: $\kw{let~rec}_{\hist}~x:\tau=e_1\kw{~in~}e_2$

\begin{itemize}
    \item Assume $G;~H\vdash \kw{let~rec}_{\hist}~x:\tau=e_1\kw{~in~}e_2:\tau$ and $S;~\sigma \vdash\kw{let~rec}_{\hist}~x:\tau=e_1\kw{~in~}e_2\xrightarrow{w}\kw{let~rec}_{\hist::v}~x:\tl{\tau_1}=e_1'\kw{~in~}e_2'$.
    \item Furthermore, we prove that $h$ consists solely of well-typed values. Specifically, we show that, after evaluation of (S-LETREC), $G;~H\vdash v:\hd{\tau_1}$ where $v$ is the last value in $h$. This ensures that $\sigma$ stays consistent with $H$
    \item Induct on the number of elements in $h$.
    \item Suppose $h=[]$. 
    \begin{itemize}
        \item By assumption, $S;~\sigma, x=[\nil]\vdash e_1\xrightarrow{v}e_1'$.
        \item By assumption and the above, $S;~\sigma, x=[v]\vdash e_2\xrightarrow{w}e_2'$.
        \item T-LETREC must be the last (non-subtype \cite{borkowski_mechanizing_2022}) rule used in the typing derivation. By inversion:
            \begin{itemize}
                \item $G;~H,x:\tau_1\vdash e_1:\tau_1$
                \item $G;~H,x:\tau_1\vdash e_2:\tau$
            \end{itemize}
        \item Since $\tau_1$ is well-formed ($G;~H\vdash \tau_1$) and $\dom{H}=\dom{\tl{H}}$, $\tl{\tau_1}$ is also well-formed, i.e. $G;~\tl{H}\vdash \tl{\tau_1}$ 
        \item By the Induction Hypothesis and Lemma \ref{lemma:expansionlemma}:
        \begin{itemize}
            \item $G;~H,x:{\tau_1}\vdash v:\hd{\tau_1}$. Thus $v$ is well-typed.
            \item $G;~\tl{H},x:\tl{\tau_1}\vdash e_1':\tl{\tau_1}$
            \item $G;~H,x:{\tau_1}\vdash w:\hd{\tau}$
            \item $G;~\tl{H},x:\tl{\tau_1}\vdash e_2':\tl{\tau}$
        \end{itemize}
        \item By (T-FBY), $G;~\tl{H}\vdash \kw{let~rec}_{v}~x:\tl{\tau_1}=e_1'\kw{~in~}e_2':\tl{\tau}$
        \item Since $v$ is the only member of $\hist$, $\hist$ contains only well-typed values, and thus the last value of $\hist$ has type $\hd{\tau_1}$ in $H$. 
    \end{itemize}

    \item Suppose $\hist \neq []$. Suppose all prior values of $h$ are well-typed.
        \begin{itemize}
            \item By assumption, $S;~\sigma, x=h::\nil\vdash e_1\xrightarrow{v} e_1'$.
            \item By assumption and the above, $S;~\sigma, x=\hist::v\vdash e_2\xrightarrow{w}e_2'$
            \item T-LETREC must be the last (non-subtype) rule used in the typing derivation. By inversion:
            \begin{itemize}
                \item $G;~H,x:\tau_1\vdash e_1:\tau_1$
                \item $G;~H,x:\tau_1\vdash e_2:\tau$
            \end{itemize}
            \item Similarly to before, $G;~\tl{H}\vdash \tl{\tau_1}$
            \item By IH and Lemma \ref{lemma:expansionlemma} (same as the $h=[]$ case):
            \begin{itemize}
                \item $G;~H,x:{\tau_1}\vdash v:\hd{\tau_1}$. Thus $v$ is well-typed.
                \item $G;~\tl{H},x:\tl{\tau_1}\vdash e_1':\tl{\tau_1}$
                \item $G;~H,x:{\tau_1}\vdash w:\hd{\tau}$
                \item $G;~\tl{H},x:\tl{\tau_1}\vdash e_2':\tl{\tau}$
            \end{itemize}
            \item By (T-FBY), $G;~\tl{H}\vdash \kw{let~rec}_{\hist::v}~x:\tl{\tau_1}=e_1'\kw{~in~}e_2':\tl{\tau}$
            \item Since the only new value added to the history is $v$ and $v$ is well-typed, $\hist::v$ remains well-typed. Since it is at the end of $\hist$, the last value of $\hist$ still has type $\hd{\tau_1}$ in $H$
        \end{itemize}
    
\end{itemize}

\textbf{Case} S-LET: Similar to S-LETREC.\\

\textbf{Case} S-FBY: $e_1\fby{} e_2$
\begin{itemize}
    \item Assume $G;~H \vdash e_1\fby{} e_2:\refinement{w}{b}{\hd{\psi_1}\land\nxt\psi_2}$ and $\sigma \vdash e_1\fby{} e_2\xrightarrow{v}\delay{e_2}$.
    \item By assumption, $S;~\sigma \vdash e_1\xrightarrow{v_1} e_1'$
    \item T-FBY must be the last (non-subtyping) rule used in the typing derivation. By inversion:
    \begin{itemize}
        \item $G;~H\vdash e_1:\refinement{w}{b}{\hd{\psi_1}}$
        \item $G;~H\vdash e_2:\refinement{w}{b}{\psi_2}$
    \end{itemize}
    \item By IH, $G;~H\vdash v_1:\refinement{w}{b}{\hd{\hd{\psi_1}}}\equiv \refinement{w}{b}{\hd{\psi_1}}$.
    \item By applying (T-DELAY) on $e_2$ and noting that $\prev{\tl{H}}(x)\preceq H(x)$ for all $x\in\dom{H}$, $\tl{H}\vdash \delay{e_2}:\refinement{w}{b}{\psi_2}$
\end{itemize}
\noindent \textbf{Case} S-DELAY: $\delay{e}$
\begin{itemize}
    \item Assume $G;~H \vdash \delay{e_1}:\tau$
    \item By assumption, $S~\prev{\sigma} \vdash e \xrightarrow{v} e'$
    \item T-DELAY must be the last (non-subtyping) rule used in the typing derivation. By inversion, $G;~\prev{H} \vdash e:\tau$
    \item By IH and using Lemma \ref{lemma:expansionlemma} to split $\tau$:
    \begin{itemize}
        \item $G;~\prev{H} \vdash v:\hd{\tau}$. Since $v$ is a value, it is a constant and thus $G;~H\vdash v:\hd{\tau}$
        \item $G;~\tl{\prev{H}} \vdash e':\tl{\tau}$
    \end{itemize}
    \item Since $\tl{\prev{H}}\equiv H$, $G;~H \vdash e':\tl{\tau}$
    \item Using (T-DELAY), $G;~\tl{H}\vdash \delay{e'}:\tl{\tau}$
\end{itemize}
\noindent\textbf{Case} S-IF-T: $\ifthenelse{x_c}{e_t}{e_e}$ 
\begin{itemize}
    \item Assume $G;~H \vdash \ifthenelse{x_c}{e_t}{e_f}:\refinement{w}{b}{\psi}$ and \\$S;~\sigma, x=\hist::\true \vdash (\kw{if~} x \kw{~then~} e_t \kw{~else~} e_f) \xrightarrow{v_t} (\kw{if~} x \kw{~then~} e_t' \kw{~else~} e_f')$
    \item By assumption:
    \begin{itemize}
        \item $S;~\sigma, x=\hist::\true \vdash e_t \xrightarrow{v_t} e_t'$
        \item $S;~\sigma, x=\hist::\true \vdash e_f \xrightarrow{v_f} e_f'$
    \end{itemize}

    \item T-IF must be the last (non-subtyping) rule used in the typing derivation. By inversion: 
        \begin{itemize}
            \item $G;~H\vdash x_c:\kw{bool}$
            \item $G;~H\vdash e_t:\refinement{w}{b}{\kw{impl}(x_c, \psi)}$
            \item $G;~H\vdash e_f:\refinement{w}{b}{\kw{impl}(\lnot x_c, \psi)}$
        \end{itemize}
    \item By IH and using Lemma \ref{lemma:expansionlemma} to split $\psi$:
        \begin{itemize}
            \item $G;~\tl{H} \vdash e_t':\tl{\refinement{w}{b}{\kw{impl}(x_c, \psi)}}$
            \item $G;~\tl{H} \vdash e_f':\tl{\refinement{w}{b}{\kw{impl}(\lnot x_c, \psi)}}$
            \item $G;~\tl{H}\vdash x_c:\tl{\kw{bool}}$ (equivalent to just (\kw{bool}))
            \item $G;~H \vdash v_t:\hd{\refinement{w}{b}{\kw{impl}(x_c, \psi)}}$. Since the current value of $x_c$ is \true, \\$\hd{\kw{impl}(x_c, \psi)}\equiv \hd{\psi}$. Therefore $G;~H \vdash v_t:\hd{\refinement{w}{b}{\psi}}$.
        \end{itemize}
    \item By (T-IF): $G;~\tl{H} \vdash (\kw{if~} x \kw{~then~} e_t' \kw{~else~} e_f'):\tl{\refinement{w}{b}{\psi}}$
\end{itemize}
\textbf{Case} S-IF-F: Symmetric to S-IF-T.\\
\noindent\textbf{Case} S-CONST: $c$ where $c$ is a constant
\begin{itemize}
    \item By Lemma \ref{lemma:expansionlemma}, $\tl{\refinement{w}{b}{\alw(w=c)}}\equiv \refinement{w}{b}{\alw(w=c)}$
    \item By Lemma \ref{lemma:expansionlemma}, $\refinement{w}{b}{\alw(w=c)}\preceq \hd{\refinement{w}{b}{\alw(w=c)}}\equiv \refinement{w}{b}{w=c}$. Therefore $G;~H \vdash c:\hd{\refinement{w}{b}{\alw(w=c)}}$
    \item Since $c$ is a constant, no free variables can occur in $c$. \\Therefore, $G;~\tl{H}\vdash c:\tl{\refinement{w}{b}{\alw(w=c)}}$
\end{itemize}
\textbf{Case} S-VAR: $x$
\begin{itemize}
   \item Assume $G;~H \vdash x:\refinement{w}{b}{\psi \land \alw(w=x)}$ and $S;~\sigma,x=\hist::v\vdash x\xrightarrow{v} x$
   \item Since the environment is non-empty (because $x$ can step), this must be a sub-term of either $\kw{let}$ or $\kw{let~rec}$. In \textbf{Case} S-LETREC, we showed that for $\sigma(x)=h$ and $H(x)=\tau$, the last value of $h$ types to $\hd{\tau}$ in $H$. Therefore, $G;~H \vdash v:\hd{\refinement{w}{b}{\psi}}$. 
   \item Because $v$ is the current value of $x$, $v=x$ (state predicate). Thus the type can be augmented as follows: $G;~H \vdash v:\hd{\refinement{w}{b}{\psi \land \alw(w=x)}}$, noting that for any $\phi$, $\hd{\alw \phi}\equiv \phi$. 
   \item By definition of $\tl{H}$, $(\tl{H})(x)=\tl{\refinement{w}{b}{\psi}}$
   \item By T-VAR, $G;~\tl{H}\vdash x:\tl{\refinement{w}{b}{\psi\land \alw (w=x)}}$
\end{itemize}
\noindent\textbf{Case} S-APP: $f\ y$
\begin{itemize}
    \item Assume $G,f:x:\tau_1 \to \tau_2 ;~H\vdash f~y:\refinement{w_2}{b_2}{\alw\phi_2[x\mapsto y]}$ and $S, f(x)=e; \sigma, y=\hist::v\vdash f(y) \xrightarrow{e[x\mapsto v]} f(y)$, where $\tau_1=\refinement{w_1}{b_1}{\phi_1}$ and $\tau_2=\refinement{w_2}{b_2}{\phi_2}$
    \item The last (non-subtyping) rule in the typing derivation must be T-APP. By inversion, $G,f:x:\tau_1 \to \tau_2;~H\vdash y:\refinement{w_1}{b_1}{\alw\phi_1}$
    \item Because $\tl{\refinement{w_1}{b_1}{\alw\phi_1}}\equiv \refinement{w_1}{b_1}{\alw\phi_1}$, then \\$G,f:x:\tau_1 \to \tau_2;~\tl{H} \vdash y:\tl{\refinement{w_1}{b_1}{\alw\phi_1}}$
    \item Using similar reasoning as in \textbf{Case} S-VAR, the last value of $y$'s history, $v$, has type $G,f:x:\tau_1 \to \tau_2;~H\vdash v:\hd{\refinement{w_1}{b_1}{\alw\phi_1}}$
    \item Using the standard substitution lemma on values extended to constants, $G,f:x:\tau_1 \to \tau_2;~H\vdash e[x\mapsto v]:\refinement{w_2}{b_2}{\phi_2[x\mapsto v]}$. Note that this is equivalent to \\$\hd{\refinement{w_2}{b_2}{\alw\phi_2[x\mapsto y]}}$ due to the correspondence between $v$ and $y$.
    \item By (T-APP) and extending the substitution lemma to type refinements that do not change over time:\\ $G,f:x:\refinement{w_1}{b_1}{\phi_1} \to \refinement{w_2}{b_2}{\phi_2};~\tl{H}\vdash f~y:\tl{\refinement{w_2}{b_2}{\alw\phi_2[x\mapsto y]}}$ which is equivalent to ${\refinement{w_2}{b_2}{\alw\phi_2[x\mapsto y]}}$
\end{itemize}
\noindent\textbf{Case} S-MODELS: $e \kw{~models~}r$

    \begin{itemize}
        \item Assume $S, \sigma \vdash e \kw{~models~} r \xrightarrow{v} e'\kw{~models~}r$. Also assume that the program is not running in robot mode (ie the program behavior is exclusively defined by the model $e$. Further assume that $G;~H\vdash e\kw{~models~}r:\tau$.
        \item By assumption, $S, \sigma \vdash e\xrightarrow{v} e'$.
        \item T-MODELS must be the last (non-subtype) rule used in the typing derivation. By inversion, $G;~H\vdash e:\tau$.
        \item By IH, $G;~\tl{H}\vdash e':\tl{\tau}$
        \item By T-MODELS, $G;~\tl{H}\vdash e'\kw{~models~}r:\tl\tau$
    \end{itemize}
\end{proof}

The interpretation of this theorem is that, if some expression satisfies a property, then evaluating that expression results in a value that satisfies that property for the single time instant in which it is used, and the remaining rewritten expression will continue to satisfy the remaining ``suffix'' of the property. Note that $e'$ is typed in a modified environment $\tl{H}$, which expresses the notion of variables updating across time instants. We prove this theorem holds in our type system via induction on the operational semantics derivation. 

\section{Implementation and Verification of a Robotics Application}
\label{sec:7}
The increasing automation of automobiles and other modes of transport sees a corresponding increase in the responsibilities of software for the safety of passengers, cargo, and infrastructure \cite{ames2014control,platzer2018logical}. As safety-critical CPS, transportation automation is a fitting application for MARVeLus. We develop our example around vehicle collision avoidance, which is responsible for maintaining safe separation of the controlled ``ego'' vehicle from other vehicles in the environment. We base our model partially on existing work formally verifying collision avoidance for adaptive cruise control in autonomous vehicles \cite{mehra2015adaptive,chen2022synchronous}, with the primary difference being the use of a discretized model for two moving vehicles.

We make the simplifying assumption that the vehicles involved move along a single one-dimensional path, as is the case for autonomous vehicles with existing lane-keeping, guided buses, or rail vehicles. Thus, the ego vehicle maintains separation by controlling its acceleration only. We also assume that the ego vehicle is traveling behind another vehicle, which is the only other participant in the scenario. The other (henceforth ``leader'') vehicle is free to take arbitrary actions, bounded only by acceleration limits. Besides externally observable properties (such as velocity and position), the ego vehicle has no \textit{a priori} knowledge of the leader's actions. The property we verify is that collision never occurs. Given that the ego vehicle starts behind the leader, this property can be expressed as $\alw(x^f<x^\ell)$; that is, the ego vehicle position $x^f$ never exceeds the leader's position $x^\ell$. Since we are primarily interested in verifying the control logic, we assume the ego vehicle has perfect knowledge of its own position and velocity, as well as that of the leader, and the vehicle instantaneously attains the commanded acceleration.

In our deterministic simulation, we implement the ``worst-case'' situation, in which the ego vehicle is accelerating at its maximum allowable rate while the leader suddenly brakes with maximum force until it completely stops. This ensures that the ego vehicle will also react safely to the less aggressive maneuvers that it is more likely to encounter in normal operation. 

We model the collision-avoidance scenario as a discrete-time system (Fig. \ref{fig:acc}), in which the leader and follower dynamics evolve at a constant time step of $dt$ seconds. At each time step, vehicle position and velocity are integrated using the standard kinematic equations of motion, while their commanded accelerations are determined by the vehicle's controller in response to the updated velocity and position values. A positive acceleration command represents the vehicle applying the throttle and thus accelerating, while a negative command represents braking at that acceleration value. The ego vehicle initially accelerates at the maximum possible rate $a_{max}$ and commands full braking acceleration $-b$ when it is no longer safe to continue accelerating as determined by the vehicles' positions and relative velocity. The leader vehicle coasts at its starting velocity for a predetermined distance, after which it commands an acceleration $-b$ which causes it to brake with maximum force until it comes to a complete stop. The ego vehicle has no knowledge of precisely when the leader will brake, and must ensure it always maintains sufficient separation to safely stop at any time. This is encoded into its controls as the relation $(x^f_{n} + s^f_n) \geq \left(x^\ell_n + \frac{(v^\ell_n-b\cdot dt)^2}{2b}+\frac{(v^\ell_n-b\cdot dt)\cdot dt}{2}\right)$, which uses the ego vehicle's anticipated stopping distance under current conditions and maximum braking $s^f$, and the assumption that the leader has already begun braking, to determine their final positions. If the ego's anticipated position reaches or exceeds that of the leader's, then it brakes until this condition no longer holds. We constrain velocities to be non-negative, to model vehicles stopping rather than accelerating backwards, when a negative acceleration is commanded while stationary. We note that the derivation of $s^f$, including its discretization adjustment, has been covered in existing work \cite{butler2011adaptive,chen2022synchronous}.

The formally-verified program consists of 89 lines of code, of which 52 are executable code (44 for core functionality and 8 for robot interfaces and logging) and 37 are additional lines for type specifications. The verified source code can be found in Appendix \ref{app:acccode}. Verification on a 2019 MacBook Pro with a 1.4 GHz Intel Core i5-8257U processor and 16 GB of memory typically takes less than 1.3 seconds.

\begin{figure}
    \centering
    \begin{align*}
        x^\ell_0&=x^\ell_{init}                   &x^\ell_{n+1}&=x^{\ell}_{n} + v^\ell_n\cdot dt\\
        v^\ell_0&=v_{init}                     &v^\ell_{n+1}&=\max(0, v^{\ell}_n + a^\ell_n\cdot dt)\\
        a^\ell_0&=0                            &a^\ell_{n+1}&=
            \begin{cases} 
            -b&\text{if } (x^\ell_{n} \geq x_{brake})\\
            0&\text{otherwise} \\
            \end{cases}\\                                 
        x^f_0&=0                            &x^f_{n+1}&=x^{f}_{n} + v^f_n\cdot dt\\
        v^f_0&=v_{init}                     &v^f_{n+1}&=\max(0, v^{f}_n + a^f_n\cdot dt)\\
        a^f_0&=a_{max}&                     a^f_{n+1}&=
            \begin{cases} 
            -b&\text{if } (x^f_{n} + s^f_n) \geq \left(x^\ell_n + \frac{(v^\ell_n-b\cdot dt)^2}{2b}+\frac{(v^\ell_n-b\cdot dt)\cdot dt}{2}\right)\\
            a_{max}&\text{otherwise} \\
            \end{cases}\\
        & &s^f_{n}&=\frac{(v^f_{n})^2}{2b} + (1+\frac{a_{max}}{b})(2v^f_n \cdot dt + \frac{4a_{max}\cdot dt^2}{2}) + \frac{v^f_{n}\cdot dt}{2}\\
       & \forall n\ .\  x^f_{n} < x^\ell_{n}
    \end{align*}
    \caption{Dynamics of the discrete-time collision avoidance system. $x^f$ and $x^\ell$ are the positions of the ego and leader vehicles, $v^\ell$ and $v^f$ are their velocities, $a^\ell$ and $a^f$ their commanded accelerations, and $s^f$ the minimum stopping distance of the ego vehicle if it were to brake at the current instant. The safety condition specifies that no collision between the vehicles at any instant.}
    \label{fig:acc}
\end{figure}

\section{Language Implementation}
\label{sec:8}
The refinement type verification is implemented as an additional pass in the Z\'elus compiler. This step follows the type checking, which is responsible for validating the base-types used in the program. The verification algorithm is built using Z3\cite{hutchison2008z3}, an SMT-solver that provides the support for proving logical propositions and searching for counterexamples.

MARVeLus extends the Z\'elus Compiler with refinement type checking for the source code. This additional step identifies refinement type annotation in the source program, translates the Z\'elus Abstract Syntax Tree (AST) to an AST with Z3 data structures, generates verification conditions (VCs) and sends VCs to the Z3 SMT solver to check for satisfiability. The compilation steps and interfaces between Z\'elus and MARVeLus are illustrated in Fig. \ref{fig:compiler}. An archive of the source code is available in the non-anonymized supplementary materials.

\begin{figure}
    \centering
    \includegraphics[width=0.7\linewidth]{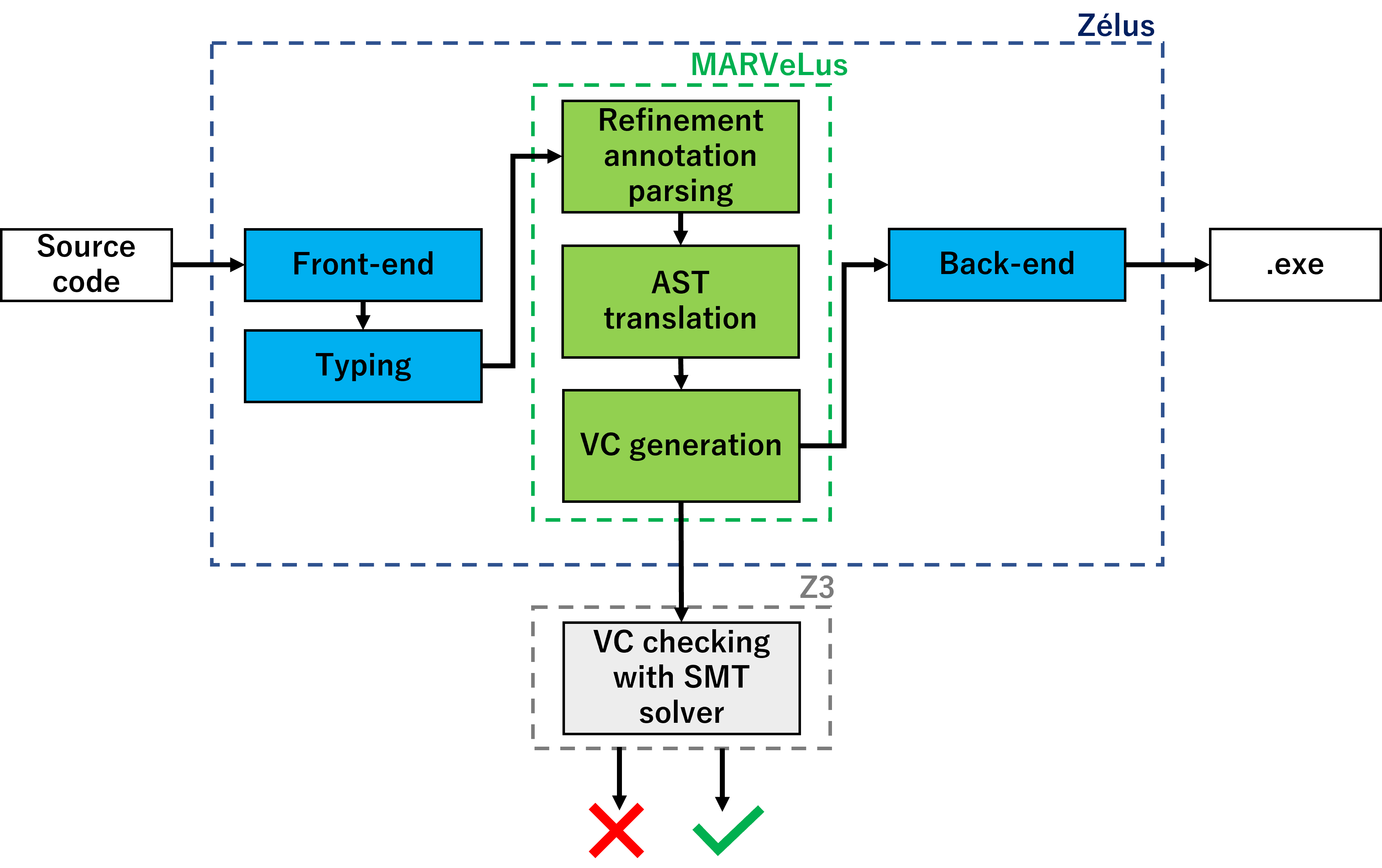}
    \caption{MARVeLus compilation steps}
    \label{fig:compiler}
\end{figure}

\subsection{Generation of Verification Conditions}

The refinement type pass in the compiler translates the Z\'elus AST enhanced with refinement types to a Z3 AST. In this translation, the program expressions, operations, patterns and equations are converted into Z3 expression objects. These objects are used to construct boolean expressions that are used as verification conditions.

The verification conditions are generated following the typing rules defined in Fig. \ref{fig:typingrules}. Each verification condition's satisfiability is checked using the Z3 SMT solver. Conditions proven to be true become powerful statements when generating other verification conditions in later parts of the program.

\subsection{Verification of Refinement Variables}

To verify a refinement typed variable declared in the form $ var: \{ w : b \ | \ \phi\} = e$, the conjunction of all expressions, $\phi_{i}$, in the current environment, $E$, must imply that the refinement $\phi$ is also true, $\bigwedge^{E} \phi_{i} \to \phi$. This arises from the (ENT-EXT) entailment rule (Fig. \ref{fig:entailment}) and the need to check that refinements assumed in the environment are consistent, so that the implication is not vacuously true. From the variable declaration, the expression $ var = e$ is added to the expression environment as an assumption.

An SMT Solver checks satisfiability by finding an example that makes the expression true. However, our goal is to prove that a verification condition is always satisfied and we need to rearrange the verification condition so the Z3 solver cannot find examples that make the verification condition false. This is achieved by the negation of the initial expression, resulting in the following verification condition: $\neg (\bigwedge^{E} \phi_{i} \to \phi )$. If this verification condition is satisfied, the refinement type is not valid and the program shows the counter-examples to the user. On the other hand, if the verification condition is unsatisfiable, the refinement expression is valid and the refinement type assumption is added to the expression environment.
\begin{lstlisting}
   let pi = 3.14159 in
   let w = 2.*.pi in
   let y: { v : float | v >= pi} = 4.0 in y
\end{lstlisting}
In the above example, {\tt pi} is a constant variable, {\tt w} is declared in terms of {\tt pi}, and {\tt y} is annotated with a refinement for which Z3 must check satisfiability. The expression environment, $E$, after variable y is declared is
$$
    E = (pi = 3.14159) \wedge (w = 2 * pi) \wedge (y = 4.0)
$$
Where the last expression corresponds to the assignment for the newly declared refinement-type variable. To show the validity of the refinement predicate, $\phi \equiv (y >= pi)$, given the variable assignment, the following verification condition is generated:
\begin{align*}
    \neg ((pi = 3.14159) \wedge (w = 2 * pi) \wedge (y = 4.0) \to (y \geq pi) )
\end{align*}
This expression is checked to be unsatisfiable by Z3, since $4 \geq 3.14159$. Therefore, the refinement type is valid and the condition $ y \geq pi$ is then added as an assumption to be used in further parts of the program.

\subsection{Verification of Refinement Functions}

A refinement function can be annotated with refined input arguments and a refined return type. In the example shown below, f is a unary refinement function of argument $ \{ v : int \ | \ v < 0\}$ and return type $\{ v : int \ | \ v >= 0 \}$. The steps to prove the refinement return type are similar to proving refinement types for variables.

In this proof, a local expression environment is created to store function body expressions and refinement-type information about the function arguments. Once the local environment is built, the function return type is used to build the verification condition that will be checked by Z3.

If the function return type is proven to be always true, it is added to a function environment. The function environment is used to check that during a function call, the arguments follow the constraints specified during the function definition.  

Once proved the function application to its argument is represented by an uninterpreted function, which is supported by Z3.

\begin{lstlisting}
let f (x: { v : float | v < 0.}) : { v : float | v >= 0. } =
    let y = x *. x in
    y
\end{lstlisting}

In the above example, function $f$ only accepts negative float inputs and the program ensures that the output is positive. The local expression environment after function declaration and before the function proof is
\begin{align*}
    E = (x < 0) \wedge (y = x * x) \wedge (v = y)
\end{align*}

Where the first condition is created from the argument refinement-type definition, the second condition is created from the function body and the third condition is a binding from the function return type to the function body return variable. For this example, Z3 checks if the verification condition $ v \geq 0$ is satisfiable using the following verification condition: 
\begin{align*}
    \neg ((x < 0) \wedge (y = x * x) \wedge (v = y) \to (v \geq 0))
\end{align*}

\section{Experimental Methods}
\label{sec:9}
\begin{figure}
    \centering
    \includegraphics[width=0.7\linewidth]{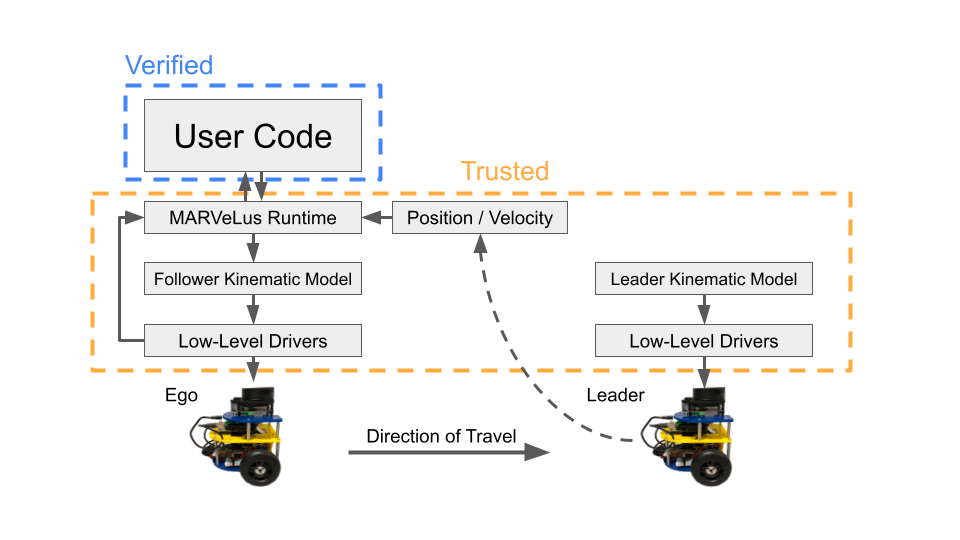}
    \caption{System Architecture for the Collision Avoidance Demonstrator}
    \label{fig:mbotstack}
\end{figure}

In addition to the aforementioned language extensions, we also made enhancements to the MARVeLus robot runtime \cite{chen2022synchronous} to support multi-robot systems. We found that the original MARVeLus robotics extensions were not easily scalable beyond a single robot, with unpredictable latency being especially problematic for synchronizing multiple robots. We retain the original software architecture, with separate code for the user MARVeLus program and runtime, the robots' kinematic models simulating the vehicles' physical attributes, and low-level drivers for sensors, actuators, and networking (Fig. \ref{fig:mbotstack}). These components all run locally on each robot. We re-implemented the communications stack to accept synchronization commands from a centralized server, ensuring timing consistency between experiments. We also selected the M-Bot Mini as our primary development target. The robots run along an \SI{8}{m} track (Fig. \ref{fig:mbotrails}). The trusted software base of our experiment consists of the MARVeLus compiler, runtime, kinematic models, and sensor and low-level drivers. That is, we trust that sensor data is always available and valid, and that actuators always attain their commanded values. We also trust that communications is instantaneous and error-free, and that data stays sufficiently consistent system-wide. Our verified control code runs on the follower vehicle, while the leader vehicle acts independently. The robots are equipped with wheel encoders to measure distance traveled and velocity, and the leader vehicle reports its position and velocity to the ego, to simulate the direct sensing that would in practice be performed using computer vision or LiDAR. The leader and ego vehicles do not otherwise communicate with one another. 

\begin{figure}
    \centering
    \includegraphics[width=0.6\linewidth]{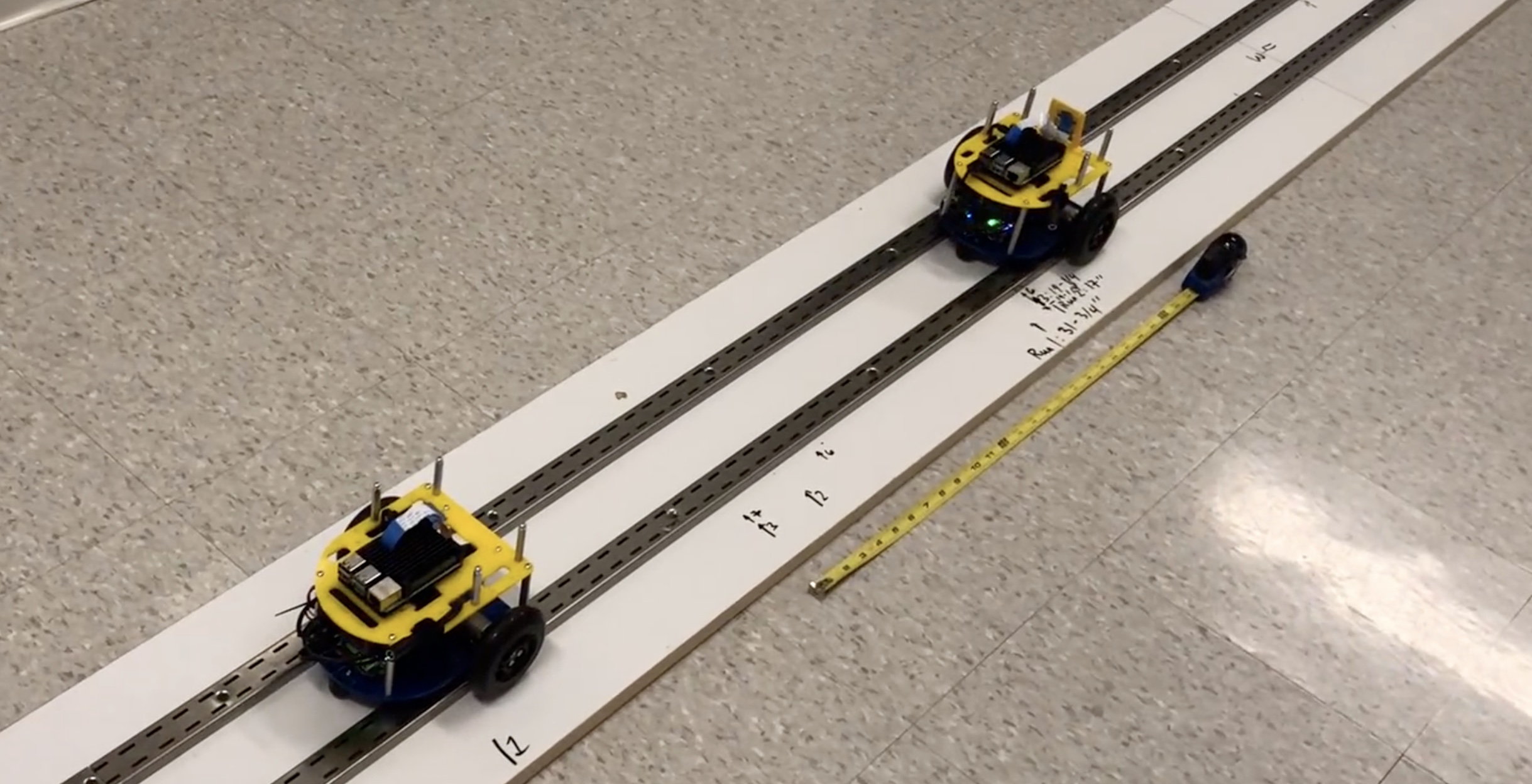}
    \caption{Ego and Leader vehicles traveling on rails}
    \label{fig:mbotrails}
\end{figure}

In addition to the collision avoidance scenario described in Section \ref{sec:7}, in which the leader brakes abruptly and stops, we tested two other variations: (1) The leader remains stationary throughout, simulating an obstacle present on the track or roadway. (2) The leader decelerates to a constant slow speed, but does not fully stop, allowing the ego vehicle to continue moving forward as long as it maintains safe separation. When both robots have stopped, or after a predetermined period of time in the case of the slow-moving leader, the final distance between the two robots is recorded.

\section{Results}
\label{sec:10}

We first ran simulations of the experiments to obtain a basis of comparison for the experimental results. We then conducted ten trials for each of the three collision avoidance scenarios described in Section \ref{sec:9}. The data obtained from one trial of the collision avoidance scenario in which the leader brakes abruptly and comes to a complete stop, is plotted in Fig. ~\ref{fig:acc2car_plot}. Data and plots from the other experiments are included in the appendix.

\begin{figure}
    \centering
    \includegraphics[width=0.7\linewidth]{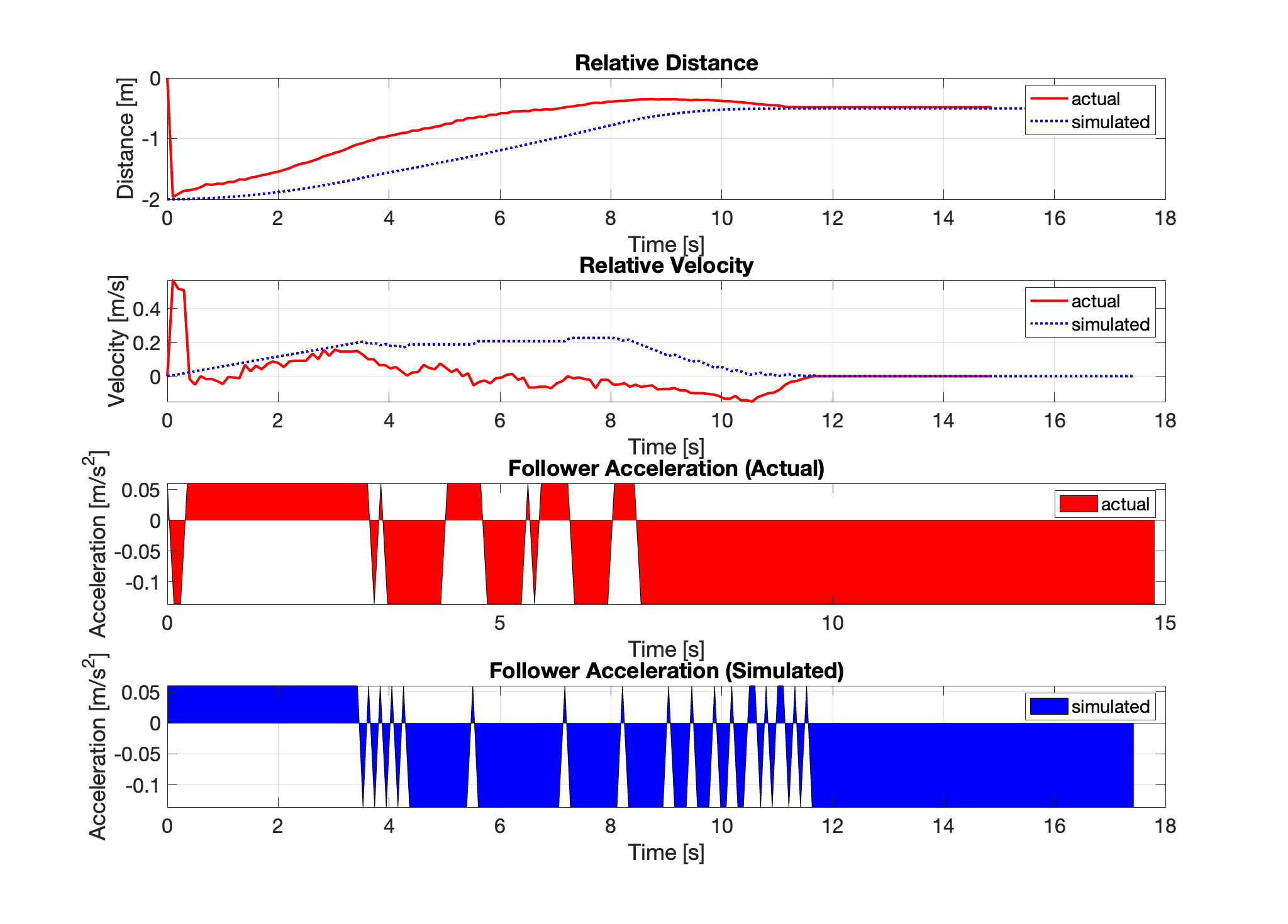}
    \caption{Relative kinematic plot for actual and simulated experiments}
    \label{fig:acc2car_plot}
\end{figure}

The kinematic data from the actual experiment differ somewhat from their simulated counterparts, which we believe arose from subtle differences in timing, initial conditions, or physical aspects of the real robots which were not accounted for in our original model. The actual commanded accelerations differed significantly from the simulation values, which is to be expected as they are reacting to different positions and velocities at any given moment. However, in spite of the variations in acceleration, the follower vehicle maintains a safe distance from the leader vehicle, allowing for the actual distance and velocity curves to converge to the simulated curves towards the end of the experiment. After finishing each of the 10 trials in each scenario, the distance between the leader and follower vehicle in their final positions was manually measured and recorded. The sample statistics for these collected distances can be seen in Table~\ref{t:ExperimentStats}.

\begin{table}[htbp!]
  \centering
  \caption{\label{t:ExperimentStats} Summary statistics of final distance measurements between vehicles}
  \begin{tabular}{l|l|l|l}
    Statistic &  Leader braking & Leader decelerating only &  Leader stationary \\
    \hline
    Mean & 0.4610 m & 1.0135 m & 0.4826 m\\
    Standard Deviation & 0.06 m & 0.4902 m & 0.0910 m\\
    Minimum & 0.3429 m & 0.7049 m & 0.2794 m\\
    Maximum & 0.5207 m & 2.1082 m & 0.5525 m
  \end{tabular}
\end{table}

\section{Related Work}
\label{sec:11}

\subsection{Verification of CPS}
Hagen and Tinelli~\cite{hagen2008scaling} developed the Kind prover for Lustre, which introduces SMT-based Lustre verification using $k$-induction, in order to automate safety proofs on synchronous dataflow programs. As with Kind, MARVeLus leverages SMT solving to inductively prove properties on synchronous programs. Although MARVeLus and Kind take different approaches to verification, a combined approach may lead to rigorous inductive proofs with reduced effort in finding invariants.

Champion et al.~\cite{champion2016kind} introduce Kind 2, an SMT-based model checker that produces invariants for synchronous programs over a bounded number of steps. For any successful proof, Kind 2 generates a proof certificate in a k-induction invariant form that can be validated by other solvers. In contrast to model-checking approach adopted by Kind 2, MARVeLus uses type properties directly to verify programs, though the end goal of reasoning on invariants remains similar. 

Kamburjan~\cite{kamburjan2021post-conditions} presents an object-oriented approach to verifying invariants in hybrid CPS. Similarly to MARVeLus, the rely-guarantee principle resembles our use of inductive invariants in discrete models, but extends this to the continuous space as well. To our knowledge, the work does not yet appear to support real-time execution of verified source code.

\subsection{Verification via Typing}
Languages such as F* \cite{swamy2016dependent} and Liquid Haskell \cite{jhala2020refinement} demonstrate the utility of allowing types to be augmented with functional or logical annotations. Both F* and Liquid Haskell interpret type annotations as constraints which are used to generate verification conditions for SMT solvers, such as Z3 \cite{hutchison2008z3}. Many of the MARVeLus type system components take cues from these two languages, lifting them into the space of synchronous programming. We chose to adopt refinement types as our theory of choice; although it is more constrained than dependent types, the quantifier-free nature of refinement predicates allows for SMT-decidable verification conditions and thus predictability \cite{jhala2020refinement}. A related application of verification via typing was demonstrated with $\Pi4$, a dependently-typed language for networking hardware.\cite{eichholz2022dependently-typed} 

\subsection{Verification and Execution} Aside from MARVeLus, VeriPhy~\cite{bohrer2018veriphy} and Koord~\cite{ghosh2020koord} also combine formal verification with execution on real systems. Although the goals of MARVeLus and Koord are similar, they differ in that Koord is intended primarily for higher-level coordination between robots, while MARVeLus is more concerned with the lower-level controls that may be specific to individual robots \cite{ghosh2020koord}. VeriPhy~\cite{bohrer2018veriphy} and similarly Melecha et al.~\cite{malecha2016towards} link a high-level abstract model of CPS to a lower-level implementation, either through runtime monitoring and switching, or verified translation. 

\subsection{Testing Autonomous Vehicles at Small-Scale}
Various experiments have previously been done on small-scale vehicles to simulate real-world autonomous vehicles \cite{rusu2010fuzzy,wen2011elman,rupp2019fast}. Rupp et al.~\cite{rupp2019fast} simulated autonomous vehicle maneuvers using scaled-down trucks and passenger cars equipped with sensors, actuators, cameras, and communication devices. ACC was implemented where two vehicles followed a lead vehicle with a constant distance or time gap while the lead vehicle tracks a given velocity profile. The experiments were done both with and without car-to-car communication. Position measurement in the platooning experiments was done visually using fiducial markers, as opposed to our simulated distance sensing.

\section{Future Work}
\label{sec:12}

\subsection{Invariants}
In its current implementation, MARVeLus simply checks typing constraints supplied by the user, relying on the user to assess the feasibility of counterexamples to refine the constraints with invariants. During program verification, one may encounter spurious counterexamples, in which an unsafe program state is found, but is not feasible in an actual system. These may be countered by supplying the verifier with invariants, as we have done manually in our example programs. In general, finding the necessary relations is not immediately obvious and requires significant additional effort by the user. There exist works that automate the process of inferring invariants \cite{hutchison2014characterizing}, including works that find invariants in synchronous programs \cite{champion2016kind}. Additionally, Liquid Haskell supports refinement inference through Horn constraint solving \cite{jhala2020refinement}, which may be adapted to take advantage of synchronous programming properties. Although orthogonal to the objectives of the current work, supporting invariant generation would help to reduce the effort involved in refining program constraints and improve user experience, particularly for complex programs.

\subsection{Trusted Components}
Although developing a single language for verification and execution holds the promise of a compact trusted base, the current implementation makes several assumptions about the safety of components interacting with the user's verified MARVeLus code. Firstly, we assume that the compiler produces a safe executable from the verified code. Formally-verified compilation of synchronous programs has been explored in other works \cite{bourke2017formally}, which allows one to directly produce executables that provably reflect their specified synchronous behavior. Combining our verification methods with a formally-verified compiler would establish a certified chain of trust in the compilation process.

Beyond compilation, we also trust the safety of MARVeLus runtime layer. For instance, we trust that network latency and throughput can be disregarded for the performance requirements of our runtime. However, this may not always be the case in real applications, particularly for resource-constrained embedded systems where the communications overhead may be nontrivial. Thus, it is necessary to characterize and verify the message-passing protocol employed for communications between the various robot components, of which there is existing work \cite{urban2015roscoq}.

We also note that, for simplicity, we chose to write the low-level control code in C, which includes feedback control for robot speed and heading. This choice was intentional, as we chose to elide some of the lower-level implementation details from our verified code example in favor of clarity. However, we note that certain components, such as motor speed control and velocity sensing, may be implemented in verified MARVeLus as well.  

For CPS operating in the real world, trust goes beyond verifying the ``cyber'' components alone. Compile-time verification implicitly assumes that system or environment assumptions hold for the real system. Error sources such as noise or external perturbations that are not properly modeled may put the system into an unsafe condition not accounted for in the specification and verification steps. There is currently no way to enforce these assumptions at runtime in MARVeLus. Runtime monitoring, perhaps based on the verified invariants in MARVeLus code could help to mitigate or notify the system of anomalies \cite{lahiri2020rtlola,bohrer2018veriphy}. Additionally, one may wish to obtain robustness measures on verification results, to determine the system's immunity to disturbances \cite{fainekos2009robustness}. 

\subsection{Richer Specifications}

Along with simulation and modeling of discrete systems, Zélus supports hybrid systems defined by differential equations, zero crossings, and labeled automata \cite{bourke2013zelus}. Presently, MARVeLus covers a subset of the discrete dynamics. Expanding the type theory to cover hybrid behavior would allow MARVeLus programs to be applied to a much greater range of CPS applications and provide more realistic models. This may be possible through the use of differential invariants as in dL \cite{platzer2008differential} or barrier certificates \cite{prajna_safety_2004}
The current specification grammar can also be expanded to include other temporal operators, leading to other properties outside of invariance to be considered as well, and in fact there is existing work exploring the semantics of temporal logic in hybrid systems \cite{hutchison_dtl2_2014}. Finally, we currently disregard numerical precision and treat most numerical values as floating-point numbers, which are interpreted internally as reals. However, this can be imprecise, especially on embedded systems that have limited support for high-precision floating-point numbers. Some research is also present in the intersection of probabilistic and synchronous programming, such as ProbZelus \cite{baudart2020reactive}, which may enable specification and analysis of stochastic processes.

\section{Conclusion}
\label{sec:13}
MARVeLus provides a method for combining refinement types with synchronous programming to formally verify cyber-physical systems directly in the source language. Since programming languages typically used for CPS verification and modeling may not be expressive enough for execution, and languages typically used to implement executable systems may not have well-defined formal semantics, it is not always clear whether a CPS implementation accurately reflects its verified counterpart. Synchronous languages are well-suited for modeling embedded systems. Thus, a verifiable and executable language gives CPS designers increased confidence in developing safety-critical systems. Part of this confidence comes from the certainty that the verification results produced are indeed correct. The extended semantics and type system we develop in this paper, along with the formal proof of type safety built upon these extensions of MARVeLus, paves the way towards building this certainty in the language. By simplifying the process of obtaining rigorous safety guarantees in CPS and providing a unified platform in which one can both verify and implement CPS, MARVeLus provides an end-to-end solution for accelerating the development of safety-critical systems and lowering the barrier for entry to formal verification.

\section*{Data Availability Statement}
Experiment data, source code, and an executable artifact are available on Zenodo \cite{chen_synchronous_2024_artifact}.

\section*{Acknowledgments}
We warmly thank Marc Pouzet and Timothy Bourke for interesting discussions, and for sharing their expertise of synchronous languages, as well as Ranjit Jhala for his detailed explanations of refinement types. We thank Tanner Slagel, Lauren White, Laura Titolo, Aaron Dutle and C\'esar Muñoz from the NASA Langley Formal Methods team, for their comments on earlier versions of this paper. We thank Peter Gaskell and the ROB 550 course staff for providing the M-Bot Mini robots used in our experiments, as well as Tigist Shiferaw and Serra Dane for their help with the experiments. Finally, we thank the reviewers for their thorough feedback and insightful comments. This research was supported in part by NSF Grant CCF-2348706.

\newpage

\appendix
\section{Collision Avoidance Source Code}\label{app:acccode}
\begin{lstlisting}[escapeinside={(*}{*)}]
let vfi:{v: float | v > 0.} = 0.5
let dt:{v: float | v > 0.} = 0.1
let b:{v: float | v > 0.} = 0.136
let xli:{v: float | v > 0.} = 2.
let xbrake:{v: float | v >= xli} = 4.
let amax:{v: float | v >= 0.} = 0.06

let node follower_control ((xf, vf, xl, vl):
  {(vxf:float)*(vvf:float)*(vxl:float)*(vvl:float) | true}) : 
  {v:float | ((((xf +. ((vf*.vf)/.(2.*.b)) +. 
    ((1.+.(amax/.b))*.(2.*.dt*.vf+.((amax*.(4.*.dt*.dt))/. 2.))) +. 
    ((vf*.dt)/. 2.)) +. 0.5
    < 
    (xl +. (((vl-.(b*.dt))*.(vl-.(b*.dt)))/.(2.*.b)) +. 
    (((vl-.(b*.dt))*.dt)/. 2.))) && (v = amax)) ||
    (((xf +. ((vf*.vf)/.(2.*.b)) +. 
    ((1.+.(amax /. b))*.(2.*.dt*.vf+.((amax*.(4.*.dt*.dt))/. 2.))) +. 
    ((vf*.dt)/. 2.)) +. 0.5 
    >= 
    (xl +. (((vl-.(b*.dt))*.(vl-.(b*.dt)))/.(2.*.b)) +. 
    (((vl-.(b*.dt))*.dt)/. 2.))) && (v = -.b)))} =
  (if ((xf +. ((vf*.vf)/.(2.*.b)) +. 
    ((1.+.(amax /. b))*.(2.*.dt*.vf+.((amax*.(4.*.dt*.dt))/. 2.)))
    +. 
    ((vf*.dt)/. 2.)) +. 0.5 
    >= 
    (xl +. (((vl-.(b*.dt))*.(vl-.(b*.dt)))/.(2.*.b)) +. 
    (((vl-.(b*.dt))*.dt)/. 2.))) 
  then (-.b) 
  else amax)
        
let node max (x:{v:float | true}) : 
  {v:float | ( x >= 0. && v = x) || (x < 0. && v = 0)} = 
  (if (x < 0.) then 0. else x)
    
let node leader_control (xl:{v:float | true}) : 
  {v:float | ((xl >= xbrake) && (v = -.b)) || 
    ((xl < xbrake) && (v = 0.))} =
  (if (xl >= xbrake) then (-.b) else 0.)
    
let node exec () : 
  {(vxf:float)*(vvf:float)*(vaf:float)*
   (vxl:float)*(vvl:float)*(vval:float) | true } =
  let rec ((xf, vf, af, xl, vl, al):
    {(x:float)*(v:float)*(a:float)*(x2:float)*(v2:float)*(a2:float) |
      (x2 > x) && (x2 >= xli) && (v2 >= 0.) && (v >= 0.) &&
      (x >= 0.) && (a >= -.b) &&
        ((((x +. ((v*.v)/.(2.*.b)) +. 
          ((1.+.(amax /. b))*.
            (2.*.dt*.v+.((amax*.(4.*.dt*.dt))/. 2.))) +. 
              ((v*.dt)/. 2.)) +. 0.5 < 
          (x2 +. (((v2-.(b*.dt))*.(v2-.(b*.dt)))/.(2.*.b)) +. 
            (((v2-.(b*.dt))*.dt)/. 2.))) && (a = amax)) ||
        (((x +. ((v*.v)/.(2.*.b)) +. 
          ((1.+.(amax /. b))*.
            (2.*.dt*.v+.((amax*.(4.*.dt*.dt))/. 2.)))  +. 
              ((v*.dt)/. 2.)) +. 0.5 >= 
          (x2 +. (((v2-.(b*.dt))*.(v2-.(b*.dt)))/.(2.*.b)) +. 
            (((v2-.(b*.dt))*.dt)/. 2.))) && (a = -.b))) &&      
        (((x2 >= xbrake) && (a2 = -.b)) || 
         ((x2 < xbrake) && (a2 = 0.))) &&
      ((x +. ((v*.v)/.(2.*.b)) +. ((v*.dt)/. 2.)) < 
      (x2 +. ((v2*.v2)/.(2.*.b)) +. ((v2*.dt)/. 2.)))}) =
    (0., vfi, amax, xli, vfi, 0.) fby
    (let vl_next = (max(vl +. (al*.dt)))
      models (robot_get ("leader_vel")) in
    (let xl_next = (xl +. (vl_next *. dt)) 
      models (robot_get ("leader_x")) in
    (let vf_next = (max(vf +. (af*.dt)))
      models (robot_get ("follower_vel")) in
    (let xf_next = (xf +. (vf_next *. dt)) 
      models (robot_get ("follower_x")) in
    (xf_next, vf_next, 
    ((if ((xf_next+.((vf_next*.vf_next)/.(2.*.b))+.
      ((1.+.(amax /. b))*.
      (2.*.dt*.vf_next+.((amax*.(4.*.dt*.dt))/. 2.)))+. 
      ((vf_next*.dt)/.2.))+.0.5 
      >= 
      (xl_next+.(((vl_next-.(b*.dt))*.(vl_next-.(b*.dt)))/.(2.*.b)) +. 
      (((vl_next-.(b*.dt))*.dt)/. 2.))) 
      then (-.b) 
      else amax)),
       xl_next, vl_next,
       ((if (xl_next >= xbrake) then (-.b) else 0.)))))))
    in  (() models robot_str ("accel", af)); (xf, vf, af, xl, vl, al)

let hybrid main () =
  let rec trigger = (period (dt)) in
    present (trigger) -> (let (x, v, a, xl, vl, al) = exec () in 
    (print_float (Timestamp.gettimeofday ()); 
    print_string ","; print_float x; print_string ","; print_float v; 
    print_string ","; print_float a; print_string ","; print_float xl; 
    print_string ","; print_float vl; print_string ","; 
    print_float al; print_newline ())) else ()
\end{lstlisting}

\newpage
\bibliographystyle{ACM-Reference-Format}
\bibliography{ref.bib}
\end{document}